\documentclass{LMCS}
\def\doi{8 (3:08) 2012}
\lmcsheading%
{\doi}
{1--28}
{}
{}
{Dec.~\phantom01, 2011}
{Aug.~10, 2012}
{}

\usepackage[english]{babel}
\usepackage[utf8]{inputenc}
\usepackage{amssymb,stmaryrd, amsmath, amsthm}
\usepackage{Guyboxes}
\usepackage{Prooftree}
\usepackage[all]{xy}
\usepackage{graphicx}
\usepackage{hyperref}
\usepackage{enumerate}

\newcommand{\restrict}{\upharpoonright}

\newcommand{\tto}{\Rightarrow}

\newcommand{\intr}[1]{\llbracket #1 \rrbracket}
\newcommand{\lpair}[1]{\langle #1 \rangle}

\newcommand{\leg}[1]{\mathcal{L}_{#1}}

\newcommand{\unfold}[1]{\widetilde{A}}

\newcommand{\Fam}{\mathrm{Fam_f}}
\newcommand{\BFam}{\mathrm{Fam}}
\newcommand{\enb}[1]{\mathrel{\vdash_{#1}}}
\newcommand{\var}{\mathtt{var}}
\newcommand{\gvar}{\mathtt{gvar}}

\newcommand{\paths}[1]{\mathcal{P}_{#1}}
\newcommand{\threads}[1]{\mathcal{T}_{#1}}
\newcommand{\prethreads}[1]{\threads{#1}'}
\newcommand{\ip}{\mathrm{ip}}
\newcommand{\jp}{\mathrm{jp}}
\newcommand{\iso}{\simeq}
\newcommand{\eval}{\Downarrow}
\newcommand{\obseq}{\cong}
\newcommand{\thread}[1]{\lceil #1 \rceil} 
\newcommand{\preleg}[1]{\leg{A}'}
\newcommand{\fst}[1]{\pi_1~#1}
\newcommand{\snd}[1]{\pi_2~#1}
\newcommand{\unit}{1}
\newcommand{\tskip}{\mathtt{()}}
\newcommand{\inj}{\iota}
\newcommand{\selim}[3]{\delta(#1, x_1 \cdot #2, x_2 \cdot #3)}
\newcommand{\id}{\mathrm{id}}
\newcommand{\new}{\mathtt{new}}
\newcommand{\mkvar}{\mathtt{mkvar}}

\newcommand{\retracts}{\lhd}

\newcommand{\cell}{\mathtt{cell}}
\newcommand{\bool}{\mathtt{bool}}
\newcommand{\nat}{\mathtt{nat}}
\newcommand{\trim}[1]{\mathrm{trim}(#1)}
\newcommand{\dst}{\mathrm{dst}}
\newcommand{\ev}{\mathrm{ev}}

\newcommand{\C}{\mathcal{C}}
\newcommand{\D}{\mathcal{D}}

\newcommand{\Gam}{\mathbf{Gam}}
\newcommand{\biggam}{\mathbf{Gam}_\infty}
\newcommand{\Vis}{\mathbf{Vis}}
\newcommand{\Inn}{\mathbf{Inn}}
\newcommand{\Path}{\mathbf{Path}}
\newcommand{\Seq}{\mathbf{Seq}}
\newcommand{\Jus}{\mathbf{Jus}}
\newcommand{\Set}{\mathbf{Set}}

\newcommand{\Lang}{\mathcal{L}}
\newcommand{\Lsums}{\mathcal{L}_{+}}

\newlength{\viewht}
\newlength{\viewlift}
\newlength{\viewdp}
\newlength{\viewdrop}

\newcommand{\pview}[1]{
\settoheight{\viewht}{\makebox{$#1$}}
\setlength{\viewlift}{\viewht}%
\addtolength{\viewlift}{-1ex}%
\raisebox{0.3\viewlift}{
  \makebox{$\ulcorner$}}
  \!#1\!
\settoheight{\viewht}{\makebox{$#1$}}
\setlength{\viewlift}{\viewht}%
\addtolength{\viewlift}{-1ex}%
\raisebox{0.3\viewlift}{
  \makebox{$\urcorner$}}
}

\begin{document}
\title[Isomorphisms of types in the presence of higher-order
  references]{Isomorphisms of types in the presence of higher-order
  references\rsuper*}

\author[P.~Clairambault]{Pierre Clairambault}
\address{Computer Laboratory, University of Cambridge, United Kingdom}
\email{pierre.clairambault@cl.cam.ac.uk}

\titlecomment{{\lsuper*}A short version of this work has appeared in the Proceedings of the 26th Annual IEEE Symposium on Logic in Computer Science (LICS), 2011}
\thanks{The author acknowledges the support of the (UK) EPSRC grant EP/HO23097 and of the Advanced Grant ECSYM of the ERC}

\begin{abstract} 
We investigate the problem of type isomorphisms in the presence of higher-order references. We first introduce a
finitary programming language with sum types and higher-order references, for which we build a fully abstract games
model following the work of Abramsky, Honda and McCusker.  Solving an open problem by Laurent, we show that two finitely
branching arenas are isomorphic if and only if they are geometrically the same, up to renaming of moves (Laurent's
forest isomorphism). We deduce from this an equational theory characterizing isomorphisms of types in our language. We
show however that Laurent's conjecture does not hold on infinitely branching arenas, yielding new non-trivial type
isomorphisms in a variant of our language with natural numbers.
\end{abstract}

\maketitle

\section{Introduction} 

During the development of denotational semantics of programming languages, there was a crucial interest in defining
models of computation satisfying particular type equations. For instance, a model of the untyped $\lambda$-calculus can
be obtained by isolating a \emph{reflexive} object (that is, an object $D$ such that $D\iso D^D$) in a cartesian closed
category. In the 80s, some people started to consider the dual problem of finding these equations that must hold in
\emph{every} model of a given language: they were coined \emph{type isomorphisms} by Bruce and Longo. In
\cite{DBLP:conf/stoc/BruceL85}, they exploited a theorem by Dezani \cite{dezani1976characterization} giving a syntactic
characterization of invertible terms in the untyped $\lambda$-calculus to prove that that the only isomorphisms of types
present in simply typed $\lambda$-calculus with respect to $\beta\eta$ equality are those induced by the equation $A\to
(B\to C) \iso B\to (A\to C)$. Later this was extended to handle such things as products
\cite{DBLP:journals/mscs/BruceCL92}, polymorphism \cite{DBLP:conf/stoc/BruceL85}, possibly with unit types
\cite{diinvertibility}, or sums \cite{DBLP:journals/apal/FioreCB06}.

The interest in type isomorphisms grew significantly when their practical impact was realized. In
\cite{DBLP:journals/jfp/Rittri91}, Rittri proposed to search functions in software libraries using their type modulo
isomorphism as a key. He also considered the possibilities offered by matching and unification of types modulo
isomorphisms \cite{DBLP:journals/ita/Rittri93}. A whole line of research has also been dedicated to the study of type
isomorphisms and their use for search tools in richer type systems (such as dependent types
\cite{DBLP:conf/fossacs/BartheP01}), along with studies about the automatic generation of the corresponding coercions
\cite{DBLP:conf/mpc/AtanassowJ04}.  Such tools were implemented for several programming languages, let us mention the
command line tool \texttt{camlsearch} written by Vouillon for CamlLight. The interested reader may refer to
the nice survey by Di Cosmo \cite{DBLP:journals/mscs/Cosmo05}.

It is worth noting that even though these tools are written for powerful programming languages featuring complex
computational effects such as higher-order references or exceptions, they rely on the theory of isomorphisms in weaker
(purely functional) languages, such as the second-order $\lambda$-calculus with pairs and unit types for
\texttt{camlsearch}. Clearly, all type isomorphisms in $\lambda$-calculus are still valid in the presence of
computational effects (indeed, the operational semantics are compatible with $\beta\eta$). What is less clear is whether
those effects allow the definition of new isomorphisms. However, it seems that syntactic methods deriving from Dezani's
theorem on invertible terms in $\lambda$-calculus cannot be extended to complex computational effects. The base setting
itself is completely different: there is no longer a canonical notion of normal form, the natural equality between terms
is no longer convertibility but observational equivalence, so new methods are required.

In \cite{DBLP:journals/mscs/Laurent05}, Laurent introduced the idea of applying game semantics to the study of type
isomorphisms (although one should mention the precursor characterization of isomorphisms by Berry and Curien
\cite{berry-curien} in the category of concrete data structures and sequential algorithms).  Exploiting his earlier work
on game semantics for polarized linear logic \cite{DBLP:journals/apal/Laurent04}, he found the theory of isomorphisms
for LLP from which he deduced (by translations) the isomorphisms for the call-by-name and call-by-value
$\lambda\mu$-calculus.  The core of his analysis is the observation that isomorphisms between arenas $A$ and $B$ in the
category $\Inn$ \cite{hyland-ong} of arenas and innocent strategies are in one-to-one correspondence with \emph{forest
isomorphisms} between $A$ and $B$, so in particular two arenas are isomorphic if and only if their representations as
forests are identical up to the renaming of vertices.

From the point of view of computational effects this looks promising, since game semantics are known to accommodate
several computational effects such as control operators \cite{laird97}, ground type
\cite{abramsky-mccusker:active-algol} or higher-order references \cite{ahm} or even concurrency
\cite{DBLP:conf/concur/Laird05} in one single framework.  Moreover, Laurent pointed out in
\cite{DBLP:journals/mscs/Laurent05} that the main part of his result, namely the fact that each $\Inn$-isomorphism
induces a forest isomorphism, does not really depend on the innocence hypothesis but only on the weaker
\emph{visibility} condition. As a consequence, his method for characterizing isomorphisms still applies to programming
languages such as Idealized Algol whose terms can be interpreted as visible strategies
\cite{abramsky-mccusker:active-algol}. Laurent raised the question whether his result could be proved without the
visibility condition, therefore yielding a characterization of isomorphisms in a programming language whose terms have
access to higher-order references and hence get interpreted as non-visible strategies \cite{ahm}.

The contributions of this paper are the following: \emph{(1)} We extend the full abstraction result in \cite{ahm} in
order to deal with sum types and the empty type, \emph{(2)} We give a new and synthetic reformulation of Laurent's tools
to approach game-theoretically the problem of type isomorphisms, \emph{(3)} We prove Laurent's conjecture in the case of
finitely branching arenas, allowing us to characterize all type isomorphisms in a finitary (integers-free) programming
language $\Lsums$ with higher-order references by the theory $\mathcal{E}$ presented\footnote{The absence of the
equation $A\to (B\to C) \iso B \to (A\to C)$ mentioned in the introduction may seem strange, but is standard in
call-by-value \cite{DBLP:journals/mscs/Laurent05} due to the restriction of the $\eta$-rule on values. Because of
call-by-value, we also have that $1$ is \emph{not} terminal, so we don't have $A \to 1 \iso 1$; instead we have the isomorphism $A \to 0 \iso 1$
up to observational equivalence.} in Figure
\ref{equational_theory}, \emph{(4)} We show however a counter-example to the conjecture when dealing with infinitely
branching arenas, and the counter-example yields a non-trivial type isomorphism in a variant of $\Lsums$ with natural
numbers. So Laurent's conjecture, in the general case, is false.

In Section \ref{section_lang} we introduce the finitary language $\Lsums$ with sums, unit types and higher-order
references, on which we define isomorphisms of types. In Section \ref{section_games}, we build a fully abstract games
model for $\Lsums$, drawing inspiration from \cite{ahm}. Then we turn to the problem of isomorphisms of types. In
Section \ref{section_isomorphisms} we first give an analysis of isomorphisms in several subcategories of the games
model, reproving and extending Laurent's theorem. Finally, we apply all of this in Section \ref{section_syntactic} to
give a characterization of isomorphisms of types in $\Lsums$ and to obtain new non-trivial isomorphisms in a variant of
$\Lsums$ with natural numbers.

\begin{figure}
\begin{eqnarray*}
A\times B &\iso_{\mathcal{E}} & B \times A\\
A\times (B\times C) &\iso_{\mathcal{E}}& (A\times B)\times C\\
\unit \times A &\iso_{\mathcal{E}}& A\\
A + B &\iso_{\mathcal{E}}& B + A\\
A + (B + C) &\iso_{\mathcal{E}}& (A + B) + C\\
0 + A &\iso_{\mathcal{E}}& A\\
A\times (B + C) &\iso_{\mathcal{E}}& A \times B + A \times C\\
(A + B) \to C &\iso_{\mathcal{E}}& (A\to C) \times (B\to C)\\
0 \to A &\iso_{\mathcal{E}}& 1\\
A \to 0 &\iso_{\mathcal{E}}& 1\\
\mathtt{var}[A] &\iso_\mathcal{E}& (A\to \unit)\times (\unit \to A)
\end{eqnarray*}
\caption{Isomorphisms in $\Lsums$}
\label{equational_theory}
\end{figure}

\section{Isomorphisms of types in $\Lsums$}
\label{section_lang}

\subsection{The language $\Lsums$}
\subsubsection{Syntax}
We introduce here a finitary variant $\Lsums$ of the programming language $\Lang$ with higher-order references modeled by Abramsky, Honda and
McCusker in \cite{ahm}: it essentially differs from $\Lang$ in the fact that the type for natural numbers has been removed. On the other hand
a sum type has been added, allowing to define all polynomial data types.
The terms and types of $\Lsums$ are defined as follows. 
\begin{eqnarray*}
A&::=& 0 ~|~ \unit~|~A+A~|~A\times A~|~A\to A~|~\mathtt{var}[A]\\\\
M &::=& x~|~\lambda x.M~|~M~M~|~\lpair{M, M}~|~\fst{M}~|~\snd{M}~|~\tskip\\
&&|~\inj_1~M~|~\inj_2~M~|~\selim{M}{N_1}{N_2}\\
&&|~\new_A~|~M:=M~|~!M~|~\mathtt{mkvar}~M~M
\end{eqnarray*}

The type annotation on $\new$ will often be omitted, whenever it is irrelevant or obvious from the context.
The typing rules for $\Lsums$ are standard, and summarized in Figure \ref{typing}. Note that in the presence of the empty type,
a term constructor is generally included as an elimination rule for $0$, along with its typing rule. We skip it here because it is
\emph{definable} : as we will see, higher-order references can be used to build an inhabitant $\bot_A: A$ for all types $A$.

\begin{figure}
\boxit{
\vspace{5pt}
\[
\prooftree
	\justifies
	\Gamma \vdash \tskip : \unit
\endprooftree
~~~~~~
\prooftree
	\Gamma \vdash M : A ~~~~~~ \Gamma \vdash N : B
	\justifies
	\Gamma \vdash \lpair{M, N} : A \times B
\endprooftree
~~~~~~
\prooftree
	\Gamma \vdash M : A \times B
	\justifies
	\Gamma \vdash \fst{M} : A
\endprooftree
\]
\vspace{5pt}
\[
\prooftree
	\Gamma \vdash M : A \times B
	\justifies
	\Gamma \vdash \snd{M} : B
\endprooftree
~~~~~~~~~~
\prooftree
	\Gamma \vdash M : A
	\justifies
	\Gamma \vdash \inj_1{M} : A + B
\endprooftree
~~~~~~~~~~
\prooftree
	\Gamma \vdash M : B
	\justifies
	\Gamma \vdash \inj_2{M} : A + B
\endprooftree
\]
\vspace{5pt}
\[
\prooftree
	\Gamma \vdash M : A + B
	~~~~
	\Gamma, x_1 : A \vdash N_1 : C
	~~~~
	\Gamma, x_2 : B \vdash N_2 : C
	\justifies
	\Gamma \vdash \selim{M}{N_1}{N_2} : C
\endprooftree
\]
\vspace{5pt}
\[
\prooftree
	\justifies
	\Gamma, x: A \vdash x: A
\endprooftree
~~~~~~~~
\prooftree
	\Gamma \vdash M : A \to B~~~~~~ \Gamma \vdash N : A
	\justifies
	\Gamma \vdash M N : B
\endprooftree
~~~~~~~~
\prooftree
	\Gamma, x : A \vdash M: B
	\justifies
	\Gamma \vdash \lambda x. M : A \to B
\endprooftree
\]
\vspace{5pt}
\[
\prooftree
        \justifies
        \Gamma \vdash \mathtt{new}_A : \mathtt{var}[A]
\endprooftree
~~~~~~~~~~~~
\prooftree
        \Gamma \vdash M:\mathtt{var}[A]
        \justifies
        \Gamma\vdash !M : A
\endprooftree
\]
\vspace{5pt}
\[
\prooftree
        \Gamma \vdash M:\mathtt{var}[A]~~~~\Gamma\vdash N:A
        \justifies
        \Gamma\vdash M:=N : \unit
\endprooftree
~~~~~~
\prooftree
        \Gamma\vdash M:A \to \unit~~~~\Gamma\vdash N:\unit\to A
        \justifies
        \Gamma\vdash \mathtt{mkvar}~M~N: \mathtt{var}[A]
\endprooftree
\]
\vspace{5pt}
}
\caption{Typing rules for $\Lsums$}
\label{typing}
\end{figure}

\subsubsection{Operational semantics}
This language is equipped with a standard big-step call-by-value operational semantics. To define it, we temporarily extend the syntax of terms with
identifiers for \textbf{locations}, denoted by $l$. Then, \textbf{values} are formed as follows:

\[
V ::= \tskip~|~\lambda x.M~|~\pi_i~V~|~\lpair{V, V}~|~l~|~\inj_i~V~|~\mathtt{mkvar}~V~V
\]

The operational semantics of $\Lsums$ are then given as an inductively generated relation $(L, s)~M \eval (L', s')~V$, where
$L$ is a (functional) set of location-type pairs, and $s$ is a partial map from locations in $L$ to values of the corresponding type,
with free locations in $L$.
By abuse of notation, we will write $l \in L$ if $(l, A) \in L$ for some type $A$. The rules are given in Figure \ref{opsem}.
Note that as usual, some store annotations are omitted when the rule considered does not affect the store. For example,
\[
\prooftree
	M \eval V~~~~ M' \eval V'
	\justifies
	M'' \eval V''
\endprooftree
\]
is an abbreviation for:
\[
\prooftree
	(L, s)~M \eval (L', s')~V~~~~~~(L', s')~M' \eval (L'', s'')~V'
	\justifies
	(L, s)~M''\eval (L'', s'')~V''
\endprooftree
\]

For a closed term $M$ without free locations, we write $M\eval$ to indicate that $(\emptyset, \emptyset)~M \eval (L, s)~V$
for some $L, s$ and $V$ (and $M\Uparrow$ to indicate that there are no such $L, s$ and $V$).
 Observational equivalence $M \obseq N$ between terms $M$ and $N$ is then defined as usual, by requiring that  
for all contexts $C[-]$ such that $C[M]$ and $C[N]$ are closed and contain no free location, $C[M]\eval$ iff $C[N]\eval$.
The corresponding equivalence relation is written $\obseq$.

\begin{figure*}
\boxit{
\vspace{5pt}
\[
\prooftree
	\justifies
	V \eval V
\endprooftree
~~~~~~~~
\prooftree
	M \eval \lpair{V_1, V_2}
	\justifies
	\fst M \eval V_1
\endprooftree
~~~~~~~~
\prooftree
	M \eval \lpair{V_1, V_2}
	\justifies
	\snd M \eval V_2
\endprooftree
~~~~~~~~
\prooftree
	M_1 \eval V_1~~~~~~M_2 \eval V_2
	\justifies
	\lpair{M_1, M_2} \eval \lpair{V_1, V_2}
\endprooftree
\]
\vspace{5pt}
\[
\prooftree
	M \eval \inj_1 V_1~~~~~~ M_1[V_1/x_1] \eval V_2
	\justifies
	\selim{M}{M_1}{M_2} \eval V_2
\endprooftree
~~~~~~~~
\prooftree
	M \eval \inj_2 V_1~~~~~~ M_2[V_1/x_2] \eval V_2
	\justifies
	\selim{M}{M_1}{M_2} \eval V_2
\endprooftree
~~~~~~~~
\prooftree
	M \eval V
	\justifies
	\inj_1 M \eval \inj_1 V
\endprooftree
\]
\vspace{5pt}
\[
\prooftree
	M \eval V
	\justifies
	\inj_2 M \eval \inj_2 V
\endprooftree
~~~~~~~~
\prooftree
	M \eval~\lambda x.M'~~~~
	N\eval~V_1~~~~
	M'[V_1/x] \eval~V_2
	\justifies
	M~N \eval~V_2
\endprooftree
\]
\vspace{5pt}
\[
\prooftree
        M_1 \eval V_1~~~~M_2\eval V_2
        \justifies
        \mathtt{mkvar}~M_1~M_2\eval \mathtt{mkvar}~V_1~V_2
\endprooftree
~~~~~~~~
\prooftree
        \justifies
        (L, s)~\mathtt{new}_A \eval (L\cup \{l:A\}, s)~l
        \using (l\not \in L)
\endprooftree
\]
\vspace{5pt}
\[
\prooftree
        (L,s)~M\eval (L', s')~l~~~~(L', s')~N\eval (L'', s'')~V
        \justifies
        (L, s)~M:= N\eval (L'', s'' \cup \{l\mapsto V\})~\tskip
\endprooftree
~~~~~~~~
\prooftree
        (L, s)~M\eval (L', s')~l
        ~~~~
        s'(l) = V
        \justifies
        (L, s)~!M \eval (L', s')~V
\endprooftree
\]
\vspace{5pt}
\[
\prooftree
        M\eval \mathtt{mkvar}~V_1~V_2
        ~~~
        N\eval V
        ~~
        V_1~V \eval \tskip
        \justifies
        M:= N \eval \tskip
\endprooftree
~~~~~~~~
\prooftree
        M\eval \mathtt{mkvar}~V_1~V_2
        ~~~~
        V_2~\tskip\eval V
        \justifies
        !M \eval V
\endprooftree
\]
\vspace{5pt}
}
\caption{Big-step operational semantics of $\Lsums$.}
\label{opsem}
\end{figure*}

\subsubsection{Syntactic extensions} In this core language, one can define all the constructs of a basic imperative programming language. For instance if $C_1$ has type
$\unit$, sequential composition $C_1; C_2$ is given by:
\[
(\lambda x:\unit.~C_2)~C_1
\]
This works only because the evaluation of $\Lsums$ is call-by-value. Likewise, a variable declaration $\mathtt{new}~x:A~\mathtt{in}~N$
(where $M$ has type $A$) can be obtained by
\[
(\lambda x:\mathtt{var}[A].~N)~\mathtt{new}_A
\]
and its initialized variant $\mathtt{new}~x=M~\mathtt{in}~N$ as expected.
As usual with general references one can define a fixed point combinator $Y_{A\to B}$ by
\[
\begin{array}{l}
\lambda f:(A\to B)\to (A\to B).\\
~\mathtt{new}~y:A\to B~\mathtt{in}\\
~~y:=\lambda a:A.~f~!y~a;\\
~~!y
\end{array}
\]
This can be easily applied to implement a $\mathtt{while}$ loop. We can also use it to build an inhabitant $\bot_A:A$ for any type $A$,
for example by $\bot_A = Y_{\unit \to A} (\lambda x.x) \tskip$.

Sum types can also be used to define datatypes. For instance, we define $\bool = 1 + 1$. It is easy to check that the usual combinators for
$\bool$ can be defined using injections and elimination of sums and that they behave in the same way \emph{w.r.t.} the operational semantics.

\subsection{Isomorphisms of types}

We are now ready to define the notion of isomorphism of types in $\Lsums$.
\begin{defi}
If $A$ and $B$ are two types of $\Lsums$, we say that $A$ and $B$ are \emph{isomorphic}, denoted by $A\iso_{\Lsums} B$,  if and only if there are two terms $x:A \vdash M:B$ and $y:B \vdash N:A$ such that:
\begin{eqnarray*}
(x:A \vdash (\lambda y.N) M) &\obseq& \id_A\\
(y:B \vdash (\lambda x.M) N) &\obseq& \id_B
\end{eqnarray*}
where $\id_A = x:A \vdash x:A$.
\end{defi}

This notion of isomorphism relies on the following notion of composition: if $x: A \vdash M: B$ and $y:B \vdash N : C$,
we define $N\circ M = x: A \vdash (\lambda y.N)~M: C$. Although we do not need it formally, let us note in passing that this composition is associative
and behaves well with respect to identities (up to observational equivalence). This can be proved directly, although reasoning on call-by-value
$\beta\eta$-reductions does not suffice --- one needs a more powerful tool such as logical relations. That this composition induces a category will
also follow directly from full abstraction since this composition coincides with composition in the games model.

\subsubsection{Isomorphisms and bad variables}

The $\mathtt{mkvar}$ construct allows to combine arbitrary ``write" and ``read" methods, forming terms
of type $\mathtt{var}[A]$ not behaving as reference cells: those are called \emph{bad variables}. We chose to include bad variables in the
language we consider for two reasons. Firstly, the games models that allow bad variables are notably simpler than those which do
not \cite{DBLP:conf/fossacs/MurawskiT09}, for which our methods do not directly apply. Secondly, the impact of allowing bad variables
on our result will be reduced by the following proposition:

\begin{prop}
Let $\Lsums'$ denote the variant of $\Lsums$ without $\mathtt{mkvar}$. Then, if $A$ and $B$ are $\var$-free types, we have $A\iso_{\Lsums} B$
if and only if $A \iso_{\Lsums'} B$.
\end{prop}
\begin{proof}
Clearly, if $A \iso_{\Lsums'} B$ we must have $A \iso_{\Lsums} B$ as well. Conversely if $A \iso_{\Lsums} B$, there are terms
$x: A \vdash M : B$ and $y: B \vdash N : A$ possibly making use of bad variables, such that $M \circ N \obseq \id_B$ and $N \circ M \obseq \id_A$.
Then, the use of bad variables in $M$ and $N$ can be eliminated.

To see how, we consider an extension $\Lsums''$ of $\Lsums$ where we add a type constructor $\gvar$ for \emph{good variables}, so that
$\Lsums''$ has both types $\var$ for bad variables and $\gvar$ for good variables. The term constructors for $\gvar$
are written $\new_g$, $!_g M$, $M:=_g N$ and obey the same rules as the corresponding constructors for $\var$;
there is no $\mkvar_g$. Then, there is translation $(-)^{t}$ from $\Lsums''$ to itself, eliminating all uses of $\mkvar$. Let us write only the
non-trivial cases:

\begin{eqnarray*}
(\var[A])^{t} &=& \gvar[A] + (A \to \unit) \times (\unit \to A)
\end{eqnarray*}
\begin{eqnarray*}
\new^{t} 		&=& \inj_1 \new_g\\
l^t			&=& \inj_1 l\\
(!M)^{t} 		&=& \selim{M^{t}}{!_g x_1}{\snd x_2~\tskip}\\
(M:=N)^{t}		&=& \selim{M^{t}}{x_1 :=_g N^{t}}{\fst x_2~N^{t}}\\
(\mkvar~M~N)^{t}	&=& \inj_2{\lpair{M^{t}, N^{t}}}
\end{eqnarray*}

In all the other cases, the translation simply goes through the term without changing it.
Likewise if $(L,s)$ is a store, $(L,s)^{t}$ is obtained by pointwise application of $(-)^{t}$. It is straightforward to prove by induction
that if $(L, s)~M \eval (L', s')~V$, then $(L, s)^t~M^{t}\eval (L', {s'})^t~V^{t}$. The converse is also true and easily provable by induction,
with the slightly stronger induction hypothesis that if $(L, s)^t~M^t\eval (L_1, s_1)~V_1$
then there exists a store $(L_0, s_0)$ and a value $V_0$ in $\Lsums''$ such that $(L_1, s_1) = (L_0, s_0)^t$, $V_1 = V_0^t$ and
$(L, s)~M\eval (L_0, s_0)~V_0$. In particular if $M$ is closed, $M\eval$ iff $M^{t} \eval$.
This translation is extended to contexts in the straightforward way, with $[]^{t} = []$, such that we always have $(C[M])^{t} = C^{t}[M^{t}]$.
Note that since $(-)^t$ does not affect $\gvar$ it is idempotent, \emph{i.e.} $t\circ t = t$.

Since $\Lsums$ is a sublanguage of $\Lsums''$, there is an obvious translation $i$ of the former to the latter. Likewise, there is
a translation $j$ from $\Lsums''$ to $\Lsums$ merging the two types for references. Overall, $j \circ t \circ i$ is a translation
from $\Lsums$ to itself whose effect is to eliminate uses of $\mkvar$, of course modifying types as a consequence. Note as well
that $j\circ i$ is the identity translation. Putting all of these together, if $M$ is a closed term of $\Lsums$, we have:

\begin{eqnarray*}
C[M] \eval 	&\Leftrightarrow& (C[M])^{t\circ i} \eval\\
		&\Leftrightarrow& C^{t\circ i}[M^{t\circ i}] \eval\\
		&\Leftrightarrow& C^{t\circ i}[M^{t\circ t \circ i}] \eval\\
		&\Leftrightarrow& C^i[M^{t\circ i}] \eval\\
		&\Leftrightarrow& C[M^{j\circ t \circ i}] \eval
\end{eqnarray*}

Therefore, if $M \obseq N$, we have $M^{j\circ t \circ i} \obseq N^{j\circ t \circ i}$. But if $M$ and $N$ are composable, it is straightforward to
check that $(N \circ M)^{j\circ t \circ i} = N^{j \circ t \circ i} \circ M^{j\circ t \circ i}$. Similarly, we have $x^{j\circ t \circ i} = x$ for any
variable $x$. From this it follows that if $M, N$ give a type isomorphism between $A$ and $B$, then $M^{j \circ t \circ i}$,
$N^{j\circ t \circ i}$ give a $\mkvar$-free type isomorphism between $A^{j \circ t \circ i}$ and $B^{j \circ t \circ i}$. But if $A$ and $B$ are
$\var$-free, we have $A^{j \circ t \circ i} = A$ and $B^{j \circ t \circ i} = B$, so we have a $\mkvar$-free isomorphism between $A$ and $B$,
thus an isomorphism in $\Lsums'$.
\end{proof}

\subsubsection{On isomorphisms without bad variables} In $\Lsums'$, when are $\var[A]$ and $\var[B]$ isomorphic? Without bad variables, there
is in general no canonical way to transform a variable of type $A$ into a variable of type $B$. It is easy to see that dereferencing
$M: \var[A]$, applying the isomorphism between $A$ and $B$ and storing the result in a new reference of type $B$ will not yield an isomorphism
because even if the language does not come with a variable equality test, it can be defined on non-trivial types.
The handling of good variables with \emph{names} in \cite{DBLP:conf/lics/MurawskiT11} suggests that to get back the original name when going
back and forth between $\var[A]$ and $\var[B]$ one has no choice but to simply forward it, which is only possible when $A$ and $B$ are syntactically
equal. Therefore we expect that a general treatment of isomorphisms with good general references would have to treat variable types as
\emph{atoms}, that you can move around but never look inside. We leave that open for future work.

\section{The games model}

\label{section_games}
We now describe the fully abstract games model of $\Lsums$, which closely follows \cite{ahm} and extends it with sums and the empty type.

\subsection{The basic category}

Our games have two players: Player (P) and Opponent (O). 
\subsubsection{Arenas}

Valid plays between $O$ and $P$ are generated by directed graphs called \emph{arenas}, which are abstract representations of types. An \textbf{arena} is
a tuple $A = \lpair{M_A, \lambda_A, I_A, \enb{A}}$ where
\begin{iteMize}{$\bullet$}
\item $M_A$ is a set of \textbf{moves},
\item $\lambda_A: M_A \to \{O, P\}\times \{Q, A\}$ is a \textbf{labeling} function which indicates whether a move is by Opponent or Player, and whether it is a Question or Answer. We write
\begin{eqnarray*}
\{O, P\}\times \{Q, A\} &=& \{OQ, OA, PQ, PA\}\\
\lambda_A &=& \lpair{\lambda^{OP}_A, \lambda^{AQ}_A}
\end{eqnarray*}
The function $\overline{\lambda_A}$ denotes $\lambda_A$ with the $O/P$ part reversed. A move $a\in M_A$ is a $O$-move (resp. $P$-move) if $\lambda^{OP}_A(a) = O$ (resp. $\lambda^{OP}_A(a) = P$).
\item $I_A\subseteq {\lambda_A}^{-1}(\{OQ\})$ is a set of \textbf{initial moves}
\item $\enb{A}\subseteq M_A^2$ is a relation called \textbf{enabling}, which satisfies that if $a \enb{A} b$, then $\lambda_A^{OP}(a) \neq \lambda_A^{OP}(b)$, and if $\lambda_A^{QA}(b) = A$ then
$\lambda_A^{QA}(a) = Q$.
\end{iteMize}
Additionally, all the arenas we consider will be \textbf{finitely branching} (for all $a\in M_A$, the set $\{m\in M_A\mid a\vdash_A m\}$ is finite).
This is crucial, since our main result relies on a counting argument.

\subsubsection{Constructions on arenas} In what follows, if $S_1$ and $S_2$ are two sets, $S_1 + S_2$ will denote their \emph{disjoint union}
defined as $\{(1, x)\mid x \in S_1\} \cup \{(2, x)\mid x \in S_2\}$. The $n$-ary variant of this operation will be written $\coprod_{i\in I} S_i$.
Whenever convenient, if $f : S_1 \to T$ and $g: S_2 \to T$ are functions, we will write $[f, g] : S_1 + S_2 \to T$ for their co-pairing,
\emph{i.e.} the function applying $f$ on elements of $S_1$ and $g$ on elements of $S_2$.

We define the \textbf{arrow} arena $A\tto B$ and the \textbf{binary product} $A \times B$:
\begin{eqnarray*}
M_{A\tto B} 		&=& M_A + M_B\\
\lambda_{A\tto B} 	&=& [\overline{\lambda_A}, \lambda_B]\\
I_{A \tto B} 		&=& \{(2,i)\mid i \in I_B\}\\
\enb{A\tto B} 		&=& \{((1, m), (1, n))\mid m \enb{A} n\} \cup \{((2, m), (2, n))\mid m \enb{B} n\}\\
			& & \cup \{((2, i_1),(1, i_2))\mid (i_1, i_2) \in I_B \times I_A\}
\end{eqnarray*}
\begin{eqnarray*}
M_{A \times B} 		&=& M_A + M_B\\
\lambda_{A\times B}	&=& [\lambda_A, \lambda_B]\\
I_{A\times B}		&=& I_A + I_B\\
\enb{A \times B}	&=& \{((1, m), (1, n))\mid m \enb{A} n\} \cup \{((2, m), (2, n))\mid m \enb{B} n\}
\end{eqnarray*}

Another construction of central importance in the model is the \textbf{lifted sum}, giving rise to a weak coproduct in $\Gam$.
If $(A_i)_{i\in I}$ is a finite family of arenas, we define:

\begin{eqnarray*}
M_{\Sigma_{i\in I} A_i}		&=& \{q\} + \{a_i\mid i\in I\} + \coprod_{i\in I} M_{A_i}\\
\lambda_{\Sigma_{i\in I} A_i}	&=& (1, q) \mapsto OQ\\
				& & (2, a_i) \mapsto PA\\
				& & (3, (i, m)) \mapsto \lambda_{A_i}(m)\\
I_{\Sigma_{i\in I} A_i}		&=& \{(1, q)\}\\
\enb{\Sigma_{i\in I} A_i}	&=& \{((1, q), (2, a_i))\mid i \in I\} \cup\\
				& & \{((2, a_i), (3,(i, m)))\mid i\in I~\&~m\in I_{A_i}\} \cup\\
				& & \{((3, (i, m)), (3, (i, n)))\mid m \enb{A_i} n\}
\end{eqnarray*}

It is obvious that these constructions preserve the fact of being finitely branching. 
The $0$-ary product (the empty arena) is denoted by $\unit$,
and will be terminal in our category.

\subsubsection{Plays}

If $A$ is an arena, a \textbf{justified sequence} over $A$ is a sequence of moves in $M_A$ together with \textbf{justification pointers}: for each non-initial move $b$, there is a pointer to
an earlier move $a$ such that $a\enb{A} b$. In this case, we say that $a$ \textbf{justifies} $b$. The transitive closure of the justification relation is called \textbf{hereditary justification}.
The relation $\sqsubseteq$ will denote the prefix ordering on justified sequences. By $s \sqsubseteq^P t$, we mean that $s$ is a $P$-ending prefix of $t$.
If $s$ is a sequence, then $|s|$ will denote its length. Moreover if $i\leq|s|$, $s_i$ will denote the $i$-th move in $s$.
A justified sequence $s$ over $A$ is a \textbf{legal play} if it is:
\begin{iteMize}{$\bullet$}
\item \textbf{Alternating}: If $s'ab \sqsubseteq s$, then $\lambda_A^{OP}(a) \neq \lambda_A^{OP}(b)$.
\item \textbf{Well-bracketed}: a question $q$ is \textbf{answered} by a later answer $a$ if $q$ justifies $a$. A justified sequence $s$ is well-bracketed if each answer is justified by the last
unanswered question, that is, the \textbf{pending} question.
\end{iteMize}
The set of all legal plays on $A$ is denoted by $\leg{A}$. We will also be interested in the set $\preleg{A}$ of well-bracketed but not necessarily alternating justified sequences on $A$, called
\textbf{pre-legal plays}. 

\subsubsection{Strategies, composition}

A \textbf{strategy} $\sigma$ on an arena $A$ (denoted $\sigma: A$) is a non-empty set of $P$-ending legal plays on $A$ satisfying \textbf{prefix-closure}, \emph{i.e.} that for all $sab \in \sigma$,
we have $s\in \sigma$ and \textbf{determinism}, \emph{i.e.} that if $sab, sac\in \sigma$, then $b=c$. As usual, strategies form a category which has
arenas as objects, and strategies $\sigma: A\tto B$ as morphisms from $A$ to $B$. If $\sigma : A\tto B$ and $\tau: B \tto C$ are strategies, their composition $\sigma; \tau: A\tto C$
is defined as usual by first defining the set of \textbf{interactions} $u\in I(A, B, C)$ of plays $u\in \leg{(A\tto B)\tto C}$ such that $u\restrict A, B\in \leg{A\tto B}$,
$u\restrict B, C \in \leg{B\tto C}$ and $u\restrict A, C\in \leg{A\tto C}$ (where $s\restrict A, B$ is the usual restriction operation essentially taking the subsequence of $s$ in $M_A$ and $M_B$,
along with the possible natural reassignment of justification pointers). The \textbf{parallel interaction} of $\sigma$ and $\tau$ is then the set 
$\sigma||\tau = \{u\in I(A, B, C) \mid u\restrict A, B \in \sigma \wedge u\restrict B, C \in \tau\}$, and the composition of $\sigma$ and $\tau$ is obtained by
the \textbf{hiding} operation, \emph{i.e.} $\sigma; \tau = \{u\restrict A, C\mid u\in \sigma||\tau\}$.
It is known (e.g. \cite{McCusker1996}) that composition is associative. It admits \emph{copycat strategies} as identities:
$\id_A = \{s\in \leg{A_1\tto A_2}\mid \forall s'\sqsubseteq^P s, s'\restrict A_1 = s'\restrict A_2\}$.

If $s\in \leg{A}$, the \textbf{current thread} of $s$, denoted $\thread{s}$, is the subsequence of $s$ consisting of all moves hereditarily justified by the same initial move as the last
move of $s$. All strategies we are interested in will be \emph{single-threaded}, \emph{i.e.} they only depend on the current thread. Formally, $\sigma: A$ is \textbf{single-threaded} if
\begin{iteMize}{$\bullet$}
\item For all $sab\in \sigma$, $b$ points in $\thread{sa}$,
\item For all $sab, t\in \sigma$ such that $ta\in \leg{A}$ and $\thread{sa} = \thread{ta}$, we have $tab\in \sigma$.
\end{iteMize}
It is straightforward to prove that single-threaded strategies are stable under composition and that $\id_A$ is single-threaded. Hence, there is a category $\Gam$ of arenas and single-threaded strategies.
The category $\Gam$ will be the base setting for our analysis. Given arenas $A$ and $B$, the arena $A\times B$ defines a cartesian product of $A$ and $B$ and the construction $A\tto B$ extends
to a right adjoint $A\times - \dashv A\tto -$, hence $\Gam$ is cartesian closed and is a model of simply typed $\lambda$-calculus. 

\subsubsection{Views, classes of strategies}

In this paper, we are mainly interested in the properties of single-threaded strategies. However, to give a complete account of the context it seems necessary to mention several classes of strategies
of interest in this setting. The most important one is certainly the class of \emph{innocent} strategies, both for historical reasons and because it is at the core of the frequent definability results -- and thus of
the full abstraction results -- in game semantics. Its definition relies on the notion of $P$-view, defined as usual by induction on plays as follows.
\[
\begin{array}{rcll}
\pview{si} &=& i & \text{if $i\in I_A$}\\
\pview{sa} &=& \pview{s}a &\text{if $\lambda_A^{OP}(a) = P$}\\
\pview{s_1 a s_2 b} &=& \pview{s_1} a b& \text{if $\lambda_A^{OP}(b) = O$ and $a$ justifies $b$}
\end{array}
\]
A strategy $\sigma: A$ is then said to be \textbf{visible} if it always points inside its $P$-view, that is, for all $sab\in \sigma$ the justifier of $b$ appears in $\pview{sa}$. The strategy $\sigma$ is 
\textbf{innocent} if it is visible, and if its behaviour only depends on the information contained in its $P$-view. More formally, whenever $sab, t\in \sigma$ such that $ta\in \leg{A}$ and $\pview{sa} = \pview{ta}$,
we must also have $tab\in \sigma$. Both visibility and innocence are stable under composition \cite{hyland-ong,abramsky-mccusker:active-algol}, thus let us denote by $\Vis$ the category of
arenas and visible single-threaded strategies and by
$\Inn$ the category of arenas and innocent strategies. Both categories inherit the cartesian closed structure of $\Gam$, but strategies in $\Inn$ are actually nothing but an abstract representation of ($\eta$-long $\beta$-normal)
$\lambda$-terms and form a fully complete model of simply-typed $\lambda$-calculus. Strategies in $\Vis$ have more freedom, they correspond in fact to programs with first-order store \cite{abramsky-mccusker:active-algol}.

\subsection{The model of $\Lsums$} 

We now show how to turn $\Gam$ into a model of $\Lsums$.

\subsubsection{Call-by-value and the $\Fam$ construction}

The three categories $\Gam$, $\Vis$ and $\Inn$ are categories of \emph{negative games} (in which Opponent always plays first), and these are known to model call-by-name computation
whereas $\Lsums$ is call-by-value.
We could have modeled it using positive games, following the lines of \cite{DBLP:journals/tcs/HondaY99}. Instead, we follow \cite{ahm}
and model $\Lsums$ in the free completion $\BFam(\Gam)$ of $\Gam$ with respect to coproducts. This will allow us to first characterize
isomorphisms in $\Gam$ (result which could be applied to a call-by-name language with state) then deduce from it the isomorphisms in $\BFam(\Gam)$.
In fact we will consider the completion $\Fam(\Gam)$ of $\Gam$ with respect to \emph{finite} coproducts, since $\Lsums$ has only finite types.

The objects of $\Fam(\Gam)$ are finite families $(A_i)_{i\in I}$ of arenas. A map from $(A_i)_{i \in I}$ to $(B_j)_{j\in J}$ is given by a
function $f: I\to J$ together with a family of strategies $(\sigma_i)_{i\in I}$ where for all $i\in I$, $\sigma_i : A_i \to B_{f(i)}$.
When $I$ is a singleton, we will write the family $(A)_{i\in I}$ simply as $\{A\}$.
Given families $A = (A_i)_{i\in I}$ and $B = (B_j)_{j\in J}$, their \emph{disjoint} sum is the family $A + B = (X_i)_{i\in I + J}$ where 
$X_{(1, i)} = A_i$, and $X_{(2, j)} = B_j$. Likewise, we define:
\begin{eqnarray*}
A\times B 	&=& (A_i \times B_j)_{(i,j) \in I\times J}\\
A\tto B 	&=& (\Pi_{i\in I} (A_i\tto B_{f(i)}))_{f \in J^I}
\end{eqnarray*}
With these definitions, $\Fam(\Gam)$ inherits a cartesian closed structure from $\Gam$. It also has coproducts given by disjoint sum of
families, let us write $\inj_1 : A \to A + B$ and $\inj_2 : B \to A + B$ the injections. As in any bicartesian closed
category the product distributes over the sum, let us write $\mathrm{d}_{\Gamma, A, B} : \Gamma\times (A + B) \to \Gamma \times A + \Gamma \times B$
for the distributivity law. By an abuse of notation we keep using $1$ for the terminal object of $\Fam(\Gam)$, that is the singleton family containing
the empty arena, we write $!_A: A\to 1$ for the terminal projection. The category $\Fam(\Gam)$ also has an initial object given by the
empty family, we denote it by $0$.

\subsubsection{Strong monad}
Moreover, the weak coproducts in $\Gam$ give rise to a \emph{strong monad} $T$ on $\Fam(\Gam)$. Its image of a family $(A_i)_{i\in I}$
is given by:
\[
TA		= \{\Sigma_{i\in I} A_i\}
\]
The unit $\eta_A$ of the monad is the family of strategies $\mathtt{in}_i : A_i \to \Sigma_{i\in I} A_i$ (the injections for the weak coproduct
structure of $\Sigma_{i\in I} A_i$) which responds to the initial Opponent move by playing $a_i$ (unless $i = 0$), then plays as copycat. The
lifting $f^* : \Gamma \times TA \to TB$ of a morphism $f : \Gamma \times A \to TB$ is given by the copairing operation of the weak coproduct,
and the distributivity law of the product over it. Using this lifting operation, there are two natural ways to define a double strength
$\dst, \dst': TA \times TB \to T(A\times B)$: $\dst$ interrogates first $TA$, whereas $\dst'$ interrogates first $TB$. The fact that
$\dst$ and $\dst'$ are distinct means that $T$ is not commutative, and the choice of preferring one or the other parallels the design
choice between left-to-right and right-to-left evaluation of a pair in a call-by-value language. Since in $\Lsums$ we have adopted
left-to-right evaluation, we prefer $\dst$ over $\dst'$.

Most of the structure of $\Lsums$ (with the exception of memory cells) can be interpreted in $\Fam(\Gam)$ following the standard interpretation of
a call-by-value language in a cartesian category with a strong monad and Kleisli exponentials \cite{moggi}:
A term $x_1: A_1, \dots, x_n : A_n \vdash M:B$ is interpreted as a morphism $\intr{M} : \Pi_{i\leq n} \intr{A_i} \to T\intr{B}$,
(where the $n$-ary product and its projections $\pi_i$ is obtained trivially by iteration of the binary product).
Details are displayed in Figure \ref{interpretation-functional}. 

\begin{figure}

{\footnotesize
\begin{eqnarray*}
\intr{\Gamma \vdash x_i : A_i} 			&=& \pi_i; \eta: \Pi_{i\leq n} \intr{A_i} \to T\intr{A_i}\\
\intr{\Gamma \vdash \lambda x.M : A\to B}	&=& \Lambda(\intr{\Gamma, x: A\vdash M: B}); \eta : \intr{\Gamma} \to T(\intr{A}\to T\intr{B})\\
\intr{\Gamma \vdash M~N}			&=& \lpair{\intr{\Gamma \vdash M: A\to B}, \intr{N: A}}; \dst; \ev^* : \intr{\Gamma} \to T\intr{B}\\
\intr{\Gamma \vdash \tskip : \unit}		&=& !_{\intr{\Gamma}}; \eta : \intr{\Gamma} \to T1\\
\intr{\Gamma \vdash \lpair{M, N} : A \times B}	&=& \lpair{\intr{\Gamma\vdash M:A}, \intr{\Gamma\vdash N:B}}; \dst\\
\intr{\Gamma \vdash \pi_i~M:A}			&=& \intr{\Gamma\vdash M: A\times B}; T\pi_i\\
\intr{\Gamma \vdash \inj_i~M: A + B}		&=& \intr{\Gamma\vdash M:A}; T(\inj_i)\\
\intr{\Gamma \vdash \selim{M}{M_1}{M_2}}	&=& \lpair{\id, \intr{M}}; (\mathrm{d}; [\intr{M_1}, \intr{M_2}])^*
\end{eqnarray*}
}

\caption{Interpretation of the pure fragment of $\Lsums$}
\label{interpretation-functional}
\end{figure}

\subsubsection{Interpretation of memory cells}

Of course we also need to give an interpretation for $\var[A]$, along with morphisms for the read and write operations of the reference cell. Once again, we follow the lines of \cite{ahm} and consider the type $\var[A]$
as the product of its read and write methods, hence we set $\intr{\var[A]} = (\intr{A} \tto T1) \times T\intr{A}$. The interpretation relies on the definition of a morphism $1 \to \intr{\var[A]}$,
that is, if $\intr{A} = \{A_i\mid i \in I\}$, a strategy $\cell: (\Pi_{i\in I} (A_i\tto 1_\bot)\times \Sigma_{i\in I} A_i)_{\bot}$, where $A_\bot = T\{A\}$ is the lift operation. Apart from the initial
protocol due to the lift, the strategy $\cell$ works by associating each read request with the latest write request and playing copycat between them. A more detailed description is given in \cite{ahm}, and an
algebraic definition is obtained in \cite{DBLP:journals/entcs/MelliesT09}. Using $\cell$ we can complete the interpretation of $\Lsums$
in $\Fam(\Gam)_T$, as displayed in Figure \ref{impure}.
The fact that $\cell$ behaves correctly is expressed by the following lemma:

\begin{figure}

{\footnotesize
\begin{eqnarray*}
\intr{\Gamma\vdash M:=N:\unit} 		&=& \lpair{\intr{\Gamma \vdash M: \var[A]}, \intr{\Gamma\vdash N:A}}; \dst; (\pi_1 \times \intr{A}; \ev)^*:\intr{\Gamma}\to T\intr{\unit}\\
\intr{\Gamma\vdash !M: A}		&=& \intr{\Gamma \vdash M: \var[A]}; \pi_2^*: \intr{\Gamma}\to T\intr{A}\\
\intr{\Gamma \vdash \mkvar~M~N: \var[A]}&=& \lpair{\intr{\Gamma\vdash M : A \to \unit}, \intr{\Gamma\vdash N : \unit \to A}}; \dst:\intr{\Gamma} \to T\intr{\var[A]}
\end{eqnarray*}
}

\caption{Interpretation of variables}
\label{impure}
\end{figure}

\begin{lem}
The equations in Figure \ref{references} hold whenever the terms concerned are well-typed.

\label{eqvar}
\end{lem}
\begin{proof}
As in \cite{ahm}, the presence of sums does not affect the proof.
\end{proof}

\begin{figure}

{\footnotesize
\begin{eqnarray*}
\intr{\Gamma \vdash \new~x:A, y:B~\mathtt{in}~M} 		&=& \intr{\Gamma \vdash \new~y:B, x:A~\mathtt{in}~M}\\
\intr{\Gamma, x:\var[A] \vdash \new~y:B~\mathtt{in}~x:=V; M} 	&=& \intr{\Gamma, x: \var[A] \vdash x:=V; \new~y:B~\mathtt{in}~M}\\
\intr{\Gamma \vdash \new~x:A, y:B~\mathtt{in}~x:=V_1; y:=V_2; M}	&=& \intr{\Gamma \vdash \new~x:A, y:B~\mathtt{in}~y:=V_2; x:=V_1; M}\\
\intr{\Gamma\vdash \new~x:A~\mathtt{in}~x:=V_1; x:=V_2; M}	&=& \intr{\Gamma \vdash \new~x:A~in~x:=V_2; M}\\
\intr{\Gamma\vdash \new~x:A~\mathtt{in}~x:=V;!x}&=& \intr{\Gamma\vdash \new~x:A~\mathtt{in}~x:=V; V}
\end{eqnarray*}
}

\caption{Equations concerning assignments and allocations}
\label{references}
\end{figure}

\subsection{Full abstraction for $\Lsums$}
\subsubsection{Soundness and adequacy}

If we have a store $(L, s)$ and a term $l_1: \var[A_1], \dots, l_n:\var[A_n] \vdash M: A$ where the $l_i$s appear in $L$ with
type $A_i$, we write $\new~L, s~\mathtt{in}~M$ as a shortcut for $\new~l_1: A_1, \dots, l_n:A_n~\mathtt{in}~l_1:=s(l_1); \dots l_n := s(l_n); M$.
Note that the order in which variables are introduced and assigned values does not matter, because of Lemma \ref{eqvar}.

\begin{prop}[Soundness]
If we have $(L, s) M \eval (L', s') V$, then for any suitably typed term $N$ we have $\intr{\new~L, s~\mathtt{in}~(\lambda x.N)~M} = 
\intr{\new~L', s'~\mathtt{in}~(\lambda x.N)~V}$.
\end{prop}
\begin{proof}
This is proved by induction on the derivation of $(L, s)~M\eval (L', s')~V$, using standard facts about bicartesian closed categories and strong
monads, along with the equations of Lemma \ref{eqvar}. 
\end{proof}

The next step is to extend the adequacy result of \cite{ahm} with sums, \emph{i.e.} that for any closed term $M$, if $\intr{M} \neq \bot$ then
$M$ converges. We do that by exploiting the retraction $A + B \retracts \mathtt{bool} \times A \times B$. Consider the language $\Lang$ of
\cite{ahm}. We define a translation of $\Lsums$ into $\Lang$ by defining $(A + B)^t = \mathtt{bool} \times A^t \times B^t$, $0^t = 1$, and
$(-)^t$ preserves all the other constructors. To extend the translation to terms, one must first note that in $\Lang$ every type has a value,
let us fix a value $V_A$ for every type $A$. Let us define the translation on terms, for the only non-trivial cases:

\begin{eqnarray*}
(\Gamma \vdash \inj_1 M : A + B)^t 		&=& \Gamma^t \vdash \lpair{\mathtt{true}, M^t, V_{B^t}} : (A+B)^t\\
(\Gamma \vdash \inj_2 M : A + B)^t 		&=& \Gamma^t \vdash \lpair{\mathtt{false}, V_{A^t}, M^t} : (A+B)^t\\
(\Gamma \vdash \selim{M}{N_1}{N_2} : C)^t	&=& \Gamma^t \vdash \mathtt{if}~\pi_1~M^t~\mathtt{then}~(\lambda x_1.N_1^t)~(\pi_2~M^t)\\
			& & \mathtt{else}~(\lambda x_2. N_2^t)~(\pi_3~M^t) : C^t
\end{eqnarray*}
The translation extends immediately to stores.
It is then a straightforward induction to prove that if $(L, s)^t~M^t \eval (L_1, s_1)~V_1$, then there exists a value $V_0$ and a store
$(L_0, s_0)$ in $\Lsums$ such that $V_1 = V_0^t$, $(L_1, s_1) = (L_0, s_0)^t$ and $(L, s)~M \eval (L_0, s_0)~V_0$. Therefore if $M^t \eval$,
$M \eval$ as well. Let us define an \textbf{embedding} of arena $\phi : A \hookrightarrow B$ as an injective function $\phi: M_A \to M_B$ 
preserving and reflecting initial moves, enabling and labelling. Likewise, there is an embedding from a family $(A_i)_{i\in I}$ to $(B_j)_{j\in J}$
if there is an injective $f: I \to J$ and for all $i\in I$ an embedding $\phi_i : A_i \hookrightarrow B_{f(i)}$.

For every type $A$ and sequent $\Gamma \vdash A$ we build an embedding $\phi_A : \intr{A} \hookrightarrow \intr{A^t}$. We detail the only non-trivial
case, \emph{i.e.} the definition of $\phi_{A+B}$. Note that if $A$ is a type in $\Lang$ and $\intr{A} = (A_i)_{i\in I}$, then
the choice of a value $\vdash V_A : A$ fixes a particular $i_0 \in I$ (such that $\intr{\vdash V_A: A}$ responds $a_{i_0}$ to the Opponent initial
move). Then, $\phi_{A+B}$ is defined from the function:
\begin{eqnarray*}
f : I + J 	&\to& (I\times J) + (I \times J)\\
(1, i)    	&\mapsto& (1, (i, j_0))\\
(2, j)		&\mapsto& (2, (i_0, j))
\end{eqnarray*}
along with the canonical embeddings of $A_i$ into $A_i \times B_{j_0}$ and of $B_j$ into $A_{i_0} \times B_j$.
This embedding $\phi : A \hookrightarrow B$ can also be applied move-by-move to plays, hence to strategies. Then, we can
prove by induction that for any term $\Gamma \vdash M : A$, we have $\phi_{\Gamma \vdash A}(\intr{M}) \subseteq \intr{M^t}$. It follows
that if $\intr{M} \neq \bot$, we have $\intr{M^t} = \phi_{\Gamma \vdash A}(\intr{M}) \neq \bot$ as well. Thus by the adequacy result
in \cite{ahm}, $M^t \eval$. Therefore, $M \eval$. We have proved:

\begin{lem}[Adequacy]
For any well typed term $M$, if $\intr{M} \neq \bot$ then $M \eval$.
\end{lem}

\subsubsection{Definability and full abstraction}

In order to get full abstraction, the main missing ingredient is definability for compact (finite) strategies.
In turn, this relies on the following factorization result:

\begin{prop}
For any arena $A$ and any finite (as a set of plays) thread-independent strategy $\sigma: 1 \tto TA$, there exist natural
numbers $k_1, k_2$ and an innocent strategy with finite view function $\tau : (\var[T1])^{k_1} \times \var[\bool^{k_2}] \tto TA$
such that:
\[
\lpair{\cell_{T1}, \dots, \cell_{T1}, \cell_{\bool^{k_2}}}; \tau = \sigma
\]
Where $A^k$ is an iterated binary product and $\lpair{\sigma_1, \dots, 
\sigma_{k+3}} = \lpair{\lpair{\sigma_1, \dots, \sigma_{k+2}}, 
\sigma_{k+3}}$.
\end{prop}
\begin{proof}
The main factorization result of \cite{ahm} gives a natural number $k_1$ and a thread-inde\-pen\-dent finite visible strategy
$\tau_1 : (\var[T1])^{k_1} \tto TA$ such that $\lpair{\cell_{T1}, \dots, \cell_{T1}}; \tau_1 = \sigma$.
The factorization theorem of \cite{abramsky-mccusker:active-algol} then allows to factorize $\tau_1$ as
$\cell_{\nat}; \tau_2$, where $\tau_2$ is an innocent strategy with finite view functions. But we have no interpretation
for $\nat$ in our model, since all arenas are supposed finitely branching! Fortunately this is not a problem: the proof works by exploiting
an injective function $code: \tau_1 \to \mathbb{N}$, encoding plays in $\tau_1$ as natural numbers and storing them
in the reference cell. However $\tau_1$ is finite, so for some $k_2 \in \mathbb{N}$ there is an encoding of $\mathrm{\tau_1}$ in $\mathbb{B}^{k_2}$,
where $\mathbb{B} = \{\mathtt{t}, \mathtt{f}\}$. Exploiting this encoding as in \cite{abramsky-mccusker:active-algol} yields the required
factorization.
\end{proof}

\begin{prop}[Definability]
Let $A$ be a type of $\Lsums$, and $\sigma: 1 \tto T\intr{A}$ a \emph{finite} strategy. Then there is a well-typed term $\vdash M : A$ of $\Lsums$
such that $\intr{M} = \sigma$.
\end{prop}
\begin{proof}
By the above factorization result, we have two natural numbers $k_1, k_2$ and an innocent strategy 
$\tau: (\var[T1])^{k_1} \times \var[\bool^{k_2}] \tto TA$ with finite view function such that 
$\sigma = \lpair{\cell_{T1}, \dots, \cell_{T1}, \cell_{\bool^{k_2}}}; \tau$.
However, recall that $\var[A]$ is just a shortcut for $(A\tto T\unit)\times TA$. Therefore we can
apply the definability result for innocent strategies of \cite{abramsky-mccusker:families} (the generalization of this result in the
presence of the empty type is straightforward), which gives a term:
\[
x_1, \dots, x_{k_1}: ((\unit \to \unit) \to \unit)\times (\unit \to \unit \to \unit), y:((\bool^{k_2} \to \unit)\times (\unit \to \bool^{k_2})) \vdash N : A
\]
With the use of bad variables, this gives $x_1, \dots, x_{k_1}: \var[\unit \to \unit], y: \var[\bool^{k_2}] \vdash N':A$ such that
$\intr{N'} = \tau_2$. Putting this together, we get $M = \new~x_1, \dots, x_{k_1}: \unit \to \unit, \new~y: \bool^{k_2}~\in~N'$ with
$\vdash M : A$, such that $\intr{M} = \sigma$.
\end{proof}

Given this we can now build the fully abstract model in a standard way, as follows. If $A$ is an arena, then the \textbf{complete} plays
on $A$ are the plays $s\in \leg{A}$ such that all questions in $s$ have been answered. We write $\sigma \obseq \tau$ the fact that
$\sigma, \tau: A$ have the same complete plays. This equivalence extends to $\Fam(\Gam)_T$: if $A$ and $B$ are families and $\sigma, \tau :A \to TB$
are morphisms $\Fam(\Gam)_T$, we write $\sigma \obseq \tau$ iff for every component $i$ of $A$, $\sigma_i \obseq \tau_i$. It is straightforward to check
that all the morphism constructions in $\Fam(\Gam)$ and $\Fam(\Gam)_T$ preserve $\obseq$, so $\Fam(\Gam)_T/\obseq$ is also a model of $\Lsums$. This does
not change the interpretation, so we still have soundness and adequacy. Putting all of this together:

\begin{thm}[Full abstraction]
The model is fully abstract, \emph{i.e.} for all $M$ and $N$ of the same type, we have $M \obseq N \iff \intr{M} \obseq \intr{N}$.
\end{thm}
\begin{proof}
$\Rightarrow$. Suppose $M \obseq N$. We can assume without loss of generality that $M$ and $N$ are closed, since $\obseq$ is a congruence (hence
stable under $\lambda$-abstraction), so $\vdash M, N : A$. Suppose $\intr{M}$ and $\intr{N}$ do not have the same complete plays, \emph{e.g.} 
$s\in \intr{M}$ but $s\not \in \intr{N}$. Then, $qsa\in \leg{\intr{A}\tto T\unit}$ where $q$ and $a$ are
respectively the question and answer in $T\unit$. Viewing $\alpha = qsa$ as a strategy, we have by definability a term 
$x: A \vdash M_\alpha : \unit$, such that $\intr{M_\alpha} = \alpha$. By construction, we have $\intr{M}; \intr{M_\alpha} \neq \bot$ and
$\intr{N}; \intr{M_\alpha} = \bot$, but $\intr{M}; \intr{M_\alpha} = \intr{(\lambda x. M_\alpha) M}$ (and similarly for $N$), thus
by adequacy $(\lambda x. M_\alpha) M \eval$ and $(\lambda x. M_\alpha) N \Uparrow$, which is absurd. Therefore $\intr{M}\obseq \intr{N}$.\\
$\Leftarrow$. Suppose $\intr{M} \obseq \intr{N}$, and take a context $C$ such that $C[M]$ is closed and $C[M] \eval$. By soundness,
$\intr{C[M]} \neq \bot$. Since $\intr{M}$ and $\intr{N}$ have the same complete plays, by immediate induction on $C$ we have $\intr{C[N]} \neq \bot$
as well. By adequacy, we have $C[N] \eval$ and $M \obseq N$.
\end{proof}

\section{Isomorphisms in $\Gam$}
\label{section_isomorphisms}
We are now going to extend Laurent's tools \cite{DBLP:journals/mscs/Laurent05} to characterize isomorphisms of types for $\Lsums$. We will first
reformulate Laurent's work in the visible and innocent cases, then extend it to characterize isomorphisms in $\Gam$.

\subsection{Isomorphisms and zig-zag strategies}

We first recall Laurent's notion of \emph{zig-zag} play.

\begin{defi}
Let $s\in \leg{A\tto B}$ be a legal play. It is \textbf{zig-zag} if 
\begin{enumerate}[(1)]
\item Each $P$-move following an $O$-move in $A$ (resp. in $B$) is in $B$ (resp. in $A$),
\item A $P$-move in $A$ immediately follows an initial $O$-move in $B$ if and only if it is justified by it, 
\item The (not necessarily legal) sequences $s\restrict A$ and $s\restrict B$ have the same pointers, \emph{i.e.} for all indices $i, j$ with
$(s\restrict A)_i$ and $(s\restrict A)_j$ defined, $(s\restrict A)_i$ points to $(s\restrict A)_j$ iff $(s\restrict B)_i$ points to
$(s\restrict B)_j$.
\end{enumerate}
If $s$ only satisfies the first two conditions, then it is \textbf{pre-zig-zag}.
\end{defi}

By extension, we will say that a strategy $\sigma$ is \textbf{pre-zig-zag} (resp. \textbf{zig-zag}) if all its plays are so. 
The core of Laurent's theorem is then that all isomorphisms in $\Vis$ are zig-zag strategies. His proof does rely on visibility, however
it only gets involved to prove that the condition $3$ of zig-zag plays
is satisfied. The first half of his argument does not use visibility and actually proves that all isomorphisms in $\Gam$ are pre-zig-zag.
Here, being mainly interested in $\Gam$, we make this explicit. We need first the following lemma.

\begin{lem}[Dual pre-zig-zag play]
Let $s\in \leg{A\tto B}$ be a pre-zig-zag play, then there exists an unique pre-zig-zag $\overline{s}\in \leg{B\tto A}$ such that $\overline{s}\restrict A = s\restrict A$
and $\overline{s}\restrict B = s\restrict B$.
\end{lem}
\begin{proof}
We define $\overline{s}$ by induction on $s$; $\overline{\epsilon} = \epsilon$, and $\overline{sab} = \overline{s}ba$. We keep the same pointers, except for the case where a move
$a$ in $A$ was justified by an initial move $b$ in $B$. Then because of the pre-zig-zag condition on $s$, $a$ is necessarily an initial move in $A$ and 
is set as the new justifier of $b$ in $\overline{s}$. There is no other possible $\overline{s}$, since the restrictions on $A$ and $B$ are constrained by the hypotheses
and their interleaving is forced by the alternation and the pre-zig-zag conditions on $\overline{s}$.
\end{proof}

\begin{lem}
If $\sigma: A\tto B$, $\tau: B\tto A$ form an isomorphism in $\Gam$, then they are pre-zig-zag and
for all $s$, $s\in \sigma \Leftrightarrow \overline{s}\in \tau$.
\end{lem}
\begin{proof}
Consider an isomorphism $\sigma: A\tto B$, $\tau: B\tto A$ in $\Gam$. We will prove by induction on even $k\in \mathbb{N}$ that all plays
of $\sigma, \tau$ whose length is less than $k$ are pre-zig-zag, and that moreover 
$\{\overline{s}\mid s \in \sigma \wedge |s|\leq k\} = \{s\in \tau\mid |s|\leq k\}$.

If $k=0$, this is trivial.
Otherwise, suppose this is true up to $k\in \mathbb{N}$, and consider $sab\in \sigma$ of length $k+2$; let us first prove condition (1).
Without loss of generality, suppose $a\in M_A$. Since $s\restrict B = \overline{s}\restrict B$, by
a straightforward zipping argument we can build an interaction $u\in I(A_1, B, A_2)$
such that $u\restrict A_1, B = s$ and $u\restrict B, A_2 = \overline{s}$, moreover since $\sigma, \tau$ form an isomorphism we
must have $u\restrict A_1, A_2 \in \id_A$. Now, we necessarily have $b\in M_B$, otherwise $u$ could be extended to $uab\in \sigma||\tau$
with $uab\restrict A_1, A_2 = (u\restrict A_1, A_2)ab$ which is not a play of the identity, contradiction. Hence $sab$ satisfies condition $1$ of
pre-zig-zag plays. 

To see why it satisfies condition $2$, take $sba\in \sigma$ with $b$ in $B$ and $a$ in $A$. If $b$ is initial in $B$, then $a$ necessarily points to it
since $\sigma$ is single-threaded. Reciprocally, suppose $a$ points to an initial move in $B$ earlier than $b$. Then we have $\overline{s}\in \tau$, and by
the same zipping argument as above we have an unique $u\in I(B_1, A, B_2)$ such that $u\restrict B_1, A = \overline{s}$ and $u\restrict A, B_2 = s$. Since
$\sigma, \tau$ form an isomorphism we also have $u\restrict B_1, B_2 \in \id_B$. Let us now extend $u$ to $u' = u b_2 a b_1$ in the unique way such that 
$u'\restrict A, B_2 = sba$ and $u'\restrict B_1, A \in \tau$. Note that we are sure that $b_1$ is a move on $B_1$ since $\overline{s} a b_1$ is a play
of $\tau$ of length $k+2$ and we already know that these satisfy the condition $1$ of pre-zig-zag plays. But we also have $u'\restrict B_1, B_2 \in \id_B$,
hence $b_2$ points in $\overline{s}$ as $b_1$ points in $s$. This means that we have $\overline{s} a b \in \tau$, such that $a$ is initial and $b$ points in
$\overline{s}$, impossible since $\tau$ is single-threaded. Hence $sba$ satisfies condition $2$ of pre-zig-zag plays.

We have proved that $sab$ is pre-zig-zag, so $\overline{sab}$ is defined. By induction hypothesis $\overline{s}\in \tau$ and the same reasoning as above
shows that it extends to $\overline{sab}\in \tau$. The argument is symmetric, 
hence $\{\overline{s}\mid s \in \sigma \wedge |s|\leq k+2\} = \{s\in \tau\mid |s|\leq k+2\}$.
\end{proof}

For the sake of completeness, let us include Laurent's argument which proves that isomorphisms in $\Vis$ are zig-zag.

\begin{lem}
If $\sigma: A\tto B$, $\tau: B\tto A$ form an isomorphism in $\Vis$, then $\sigma$ and $\tau$ are zig-zag strategies.
\end{lem}
\begin{proof}
We already know that $\sigma$ and $\tau$ are pre-zig-zag strategies. We show by induction on $n\in \mathbb{N}$ that for all $s\in \sigma$, if $|s|\leq n$
then $s\restrict A$ and $s\restrict B$ have the same pointers. Take now $s\in \sigma$, and $sab \in \sigma$, suppose w.l.o.g. that $a\in M_A$. Suppose
$a$ points to $(s\restrict A)_i$, then $b$ points to $(s\restrict B)_i$. Indeed, it cannot point to $(s\restrict B)_j$ with $j>i$ since that would
break visibility for $\sigma$. But if it points to $(s\restrict B)_j$ with $j<i$ we use the same reasoning on the dual pre-zig-zag play $\overline{sab}$
and get a contradiction with the fact that $\tau$ is visible.
\end{proof}

Let us denote by $\Gam_i$, $\Vis_i$ and $\Inn_i$ the groupoids of arenas and isomorphisms on the respective categories.
In the next sections, we use these facts to give more combinatorial representations of $\Gam_i$, $\Vis_i$ and $\Inn_i$.

\subsection{Notions of game morphisms}

Laurent's isomorphism theorem works by relating isomorphisms in $\Gam$ with isomorphisms in a simpler category which has arenas as objects and
\emph{forest morphisms}, \emph{i.e.} maps on moves that preserve initiality
and enabling. Relaxing the visibility conditions requires us to also consider relaxed notions of game morphisms, that we present here.

In what follows we will make use of the \textbf{prefix functions} $\ip$ and $\jp$ on justified sequences, defined by $\ip(\epsilon) = \epsilon$
and $\ip(sa) = s$, and $\jp(si) = \epsilon$ if $i$ does not have a pointer, $\jp(s_1 a s_2 b) = s_1 a$ if $b$ points to $a$.

\begin{defi}
Let $A$ be an arena. A \textbf{path} on $A$ is a play $s\in \leg{A}$ such that except for the initial move, every move in $s$ points to the previous move.
Formally, for all $s'ab\sqsubseteq s$, $a$ justifies $b$ in $s$. Let $\paths{A}$ denote the set of paths on $A$. A \textbf{path morphism} from $A$ to $B$
is a function $\phi:\paths{A} \to \paths{B}$ such that $\ip\circ \phi = \phi\circ \ip$ and which preserves $Q/A$ labeling: for all $sa\in \paths{A}$ with
$\phi(sa) = \phi(s) b$, we have $\lambda_A^{QA}(a) = \lambda_B^{QA}(b)$. There is a category $\Path$ of arenas and path morphisms.
\end{defi}

This category $\Path$ comes with its own notion of isomorphisms of arenas. Note that whenever $A$ is a forest, this is exactly Laurent's notion
of forest isomorphism. We now introduce two weaker notions of morphisms for arenas. In what follows, let us call a legal play on $A$ with only
one initial move a \textbf{thread} on $A$, and denote the set of threads on $A$ by $\threads{A}$. Likewise, 
let us call a pre-legal play with one initial move a pre-legal thread and let us denote these by $\prethreads{A}$.

\begin{defi}
Let $A$, $B$ be arenas, and let $\phi : \prethreads{A} \to \prethreads{B}$
We say that $\phi$ is a \textbf{sequential morphism}
from $A$ to $B$ if $\ip\circ \phi = \phi\circ \ip$, and if it preserves $Q/A$ labeling, \emph{i.e.} for all $\phi(sa) = \phi(s) b$ we have $\lambda_A^{QA}(a) = \lambda_B^{QA}(b)$.
We say that it is a \textbf{justified morphism} if, additionally, $\jp\circ \phi = \phi \circ \jp$.
There are two categories $\Seq$ of arenas and sequential morphisms and $\Jus$ of arenas and
justified morphisms.
\end{defi}

The condition on sequential morphisms amounts to the fact that they preserve play extension, \emph{i.e.} for all pre-legal threads $sa\in \prethreads{A}$,
$\phi(sa)$ must be an immediate extension of $\phi(s)$. In other words, a sequential morphism preserves the forest structure of the set of
pre-legal threads given by the prefix ordering. However it does not have to preserve pointers : it could for instance send a play
$\xymatrix@C=5pt{\circ\ar@{-}[r]&\bullet \ar@{-}[r]&\circ \ar@{-}[r]&\bullet}$ to $\xymatrix@C=5pt{\circ\ar@{-}[r]&\bullet \ar@{-}[r]&\circ
&\bullet\ar@{-}@/_/[lll]}$, where occurrences of $\circ$ and $\bullet$ are respectively Opponent and Player moves. These weak forms
of morphisms will play an important role in the subsequent development as they have a close relationship to isomorphisms in $\Gam$. Justified
morphisms are those sequential morphisms which additionally preserve pointers: those will appear to be in relationship with isomorphisms
in $\Vis$.

As above, we will denote by $\Seq_i$, $\Jus_i$ and $\Path_i$ the groupoids of invertible maps in $\Seq$, $\Jus$ and $\Path$. These groupoids
will soon appear to be identical to $\Gam_i$, $\Vis_i$ and $\Inn_i$. To prove this, we need the following lemma.

\begin{lem}
Let $s\in \prethreads{A}$, and $\sigma:A\tto B$ an isomorphism in $\Gam$. There is then an unique play $s'\in \sigma$ such that $s'\restrict A = s$.
\end{lem}
\begin{proof}
Remark first that if $\sigma: A\tto B$ and $\tau: B \tto A$ are inverses then they are both total, \emph{i.e.} for all $s\in \sigma$ and $sa\in \leg{A\tto B}$
there must be $b$ such that $sab\in \sigma$, assuming it is not the case easily leads to a contradiction. We now prove the lemma by induction on $s$. If $s = \epsilon$,
this is trivial. Otherwise, suppose $sa\in \prethreads{A}$ and we have by induction hypothesis $s'\in \sigma$ such that $s'\restrict A = s$. If $a$ is a
$P$-move in $A$ (hence an $O$-move in $A\tto B$), there is an unique $b$ such that $s'ab\in \sigma$, and we do have $s'ab\restrict A = sa$. If
$a$ is an $O$-move in $A$ (hence a $P$-move in $A\tto B$), then let $\tau: B\tto A$ be the inverse of $\sigma$, since $s'\in \sigma$ we have 
$\overline{s'}\in \tau$. Being part of an isomorphism $\tau$ is total, hence there is $b$
such that $\overline{s'}ab\in \tau$. We deduce from this that $s' b a\in \sigma$, and we have $s'ba\restrict A = sa$ as needed. This choice is unique:
if there is another play $t\in \sigma$ such that $t\restrict A = sa$, then $t = t'b'a$ (since $t$ is zig-zag). By induction
hypothesis we have $t' = s'$, thus $s'b'a\in \sigma$. From this we deduce that $\overline{s'}ab'\in \tau$, so $b=b'$ by determinism of $\tau$.
\end{proof}

\begin{prop}
If $\C \iso \D$ means that two groupoids $\C$ and $\D$ are \emph{isomorphic}, then we have:
\begin{eqnarray*}
\Gam_i &\iso& \Seq_i\\
\Vis_i &\iso& \Jus_i
\end{eqnarray*}
\end{prop}
\begin{proof}
Let us first define a functor $F: \Gam_i \to \Seq_i$. It is defined as the identity on arenas. Let $\sigma : A\tto B$ be an isomorphism, and
let $s\in \prethreads{A}$ then we define $\phi_\sigma(s) = s'\restrict B$, where $s'$ is the unique play on $A\tto B$ which existence is ensured by the lemma above.
The function $\phi_\sigma$ commutes with $\ip$ since $\sigma$ is a pre-zig-zag strategy. To any question it cannot associate an answer, as that would immediately break well-bracketing on $\sigma$.
But to any answer it cannot associate a question, as that would immediately break well-bracketing on $\sigma^{-1}$.
Then we define $F(\sigma) = \phi_\sigma$. It is obvious that $F$ preserves identities and composition\footnote{In fact, this construction can be seen as a particular case
of Hyland and Schalk's faithful functor from games to relations \cite{hyland-schalk}, where the relation happens to be functional.}.

Reciprocally, suppose $\phi: A\to B$ is a sequential isomorphism. We mimic the usual definition of the identity by setting
$G(\phi) = \{s\in \leg{A\tto B}\mid \forall s'\sqsubseteq^P s,~ \phi(s'\restrict A) = s'\restrict B\}$ (We apply $\phi$ on
plays whereas it is normally only defined on \emph{threads}, however it can be canonically extended to plays, so this is not ambiguous).
It is obvious that this construction is functorial, and that it is inverse to $F$.

We have now an isomorphism $\Gam_i \iso \Seq_i$ which restricts naturally to $\Vis_i$ and $\Jus_i$. Indeed if $\sigma: A\tto B$ is a visible isomorphism,
it is a zig-zag strategy therefore $s\in \prethreads{A}$ and $\phi_\sigma(s)$ have the same pointers, which means that $\jp\circ \phi_\sigma = \phi_\sigma \circ \jp$.
Reciprocally if $\phi_\sigma$ is a justified morphism, all $s\in \sigma$ must be such that $s\restrict A$ and $s\restrict B$ have the same pointers,
therefore $\sigma$, being pre-zig-zag, always points in its $P$-view.
\end{proof}

\subsection{Innocent and visible case}
\label{laurentstheorem}

In this section, we use the framework described above to recall Laurent's results. We have proved above that isomorphisms in $\Vis$ correspond
to isomorphisms in $\Jus$, which we are now going to compare with isomorphisms in $\Path$. 

\begin{lem}
There is a full functor $H: \Vis_i \to \Path_i$.
\end{lem}
\begin{proof}
We have built in the above section a full and faithful functor (actually an isomorphism) $F: \Vis_i \to \Jus_i$. From a visible isomorphism $\sigma : A\tto B$
we set $H(\sigma) = F(\sigma) \restrict \paths{A}$, where $f \restrict E'$ restricts a function $f: E\to F$ to a subset $E'\subseteq E$ of its domain.
The image of a path by $F(\sigma)$ is always a path since it is a justified morphism, hence $H(\sigma) : \paths{A} \to \paths{B}$.

To see why $H$ is full, suppose we have a path morphism $\phi: \paths{A} \to \paths{B}$. Then $\phi$ admits a canonical extension $\phi^*: \prethreads{A} \to \prethreads{B}$. To
define $\phi^*(s)$ we reason by induction on $s$, and set $\phi^*(\epsilon) = \epsilon$ and $\phi^*(sa) = \phi^*(s) a'$, where $a'$ is the last move of $\phi(p_a)$, $p_a$
being the path of $a$ in $s$. The move $a'$ keeps the same pointer as $a$. It is clear that this defines as needed a justified morphism $\phi^*$ such that $H(\phi^*) = \phi$.
\end{proof}

This ensures that arenas $A$ and $B$ are isomorphic in $\Vis$ if and only if they are isomorphic in $\Path$, \emph{i.e.} they are geometrically the same. 
Let us mention that as Laurent proved, this correspondence is one-to-one in the innocent case: one can prove that there is only one innocent zig-zag strategy corresponding
to a particular path isomorphism, hence $H$ restricts to an isomorphism of groupoids $H': \Inn_i \to \Path_i$.

\begin{exa}
Note that $H$ itself is \emph{not} faithful: we can exploit non-innocence to build
non-uniform isomorphisms, \emph{i.e.} isomorphisms which change their underlying path isomorphism as the interaction progresses. For an example,
consider the arena 
\[
A = \raisebox{20pt}{
\xymatrix@R=2pt{
&&q\\
q_1\ar@{-}@/^/[urr]&q_2\ar@{-}@/^/[ur]&a\ar@{-}[u]\\
a\ar@{-}[u]&a\ar@{-}[u]
}}
\]
which is the interpretation of $(\mathtt{bool}\to \unit) \to \unit$ in call-by-value and of $\unit\times \unit \to \unit$ in
call-by-name. Consider now the strategy $i:A\tto A$ which behaves as follows. It starts by playing as the identity on $A$. The first time Opponent
plays $q_1$ or $q_2$ on the left hand side, it simply copies it. Starting from the second time Opponent plays $q_1$ or $q_2$ though, it swaps them. An example
play of $i$ is given in Figure \ref{involution}. Although it is not the identity, $i$ is its own inverse. Its image by $H$ only takes into account the first behaviour
or $i$, thus is the same as for $\id_A$: the identity path morphism on $A$. From this strategy we can extract the following term $f:B \vdash M:B$ of $\Lsums$, where
$B = (\mathtt{bool}\to \unit) \to \unit$.
\[
\begin{array}{l}
\mathtt{new}~r:=\mathtt{true}~\mathtt{in}\\
\lambda g. f (\lambda b. \mathtt{if}~!r~\mathtt{then}~r:=\mathtt{false};g~b~\mathtt{else}~g~(\mathtt{not}~b))
\end{array}
\]

\begin{figure}
\[
\xymatrix@R=0pt{
			&A\ar@{=>}[rrr]		&			&			&A				&\\
			&			&			&			&				&q\\
			&			&q\ar@{-}@/^/[urrr]	&			&				&\\
q_1\ar@{-}@/^/[urr]	&			&			&			&				&\\
			&			&			&q_1\ar@{-}@/^/[uuurr]	&				&\\
q_1\ar@{-}@/^/[uuurr]	&			&			&			&				&\\
			&			&			&			&q_2\ar@{-}@/^/[uuuuur]		&\\
			&q_2\ar@{-}@/^/[uuuuur]	&			&			&				&\\
			&			&			&q_1\ar@{-}@/^/[uuuuuuurr]&				&
}
\]
\caption{A play of the non-trivial involution $i$ on $A$}
\label{involution}
\end{figure}
Although $M$ is not the identity it is an involution on $B$, \emph{i.e.} we have $(\lambda f. M) (M x) \obseq_{\Lsums} x$. Such non-trivial involutions
cannot be defined using only purely functional behaviour.
\end{exa}

We give in Figure \ref{groupoids} a summary of all the groupoids of isomorphisms encountered so far, along with their relations. Following it,
the question of finding the isomorphisms in $\Gam$ boils down to the definition of an arrow from $\Seq_i$ to $\Path_i$ in this diagram, which is what we
will attempt in the next two subsections.

\begin{figure}
\[
\xymatrix@C=30pt@R=30pt{
\Inn_i  \ar[rr]^{\text{iso}}
        \ar[d]^{\subseteq}      &&\Path_i\\
\Vis_i  \ar[r]^{\text{iso}}
        \ar[d]^{\subseteq}      &\Jus_i \ar[ur]^{\text{full}}_{\text{(not faithful)}}\ar[dr]^{\subseteq}\\
\Gam_i  \ar[rr]^{\text{iso}}&&\Seq_i    \ar@{.>}[uu]_{?}
}
\]
\caption{Relations between all groupoids of isomorphisms}
\label{groupoids}
\end{figure}

\subsection{Non-visible isomorphisms by counting}

We have seen above that we can build a full functor $\Vis_i \to \Path_i$, which allows to characterize isomorphic arenas in $\Vis$. However,
this construction relies heavily on visibility. We now investigate how to get rid of it and prove that two arenas $A$ and $B$ are isomorphic in $\Gam$ if and only if they are isomorphic in $\Path$.
In this subsection, we will describe for pedagogical reasons an intuitive approach to the proof, which relies on counting. However this approach
suffers from some defects, hence the full proof (described in the next subsection) will follow slightly different lines.

If $a\in M_A$, let us call its \textbf{arity}
the quantity $ar(a) = |\{m\in M_A\mid a \enb{A} m\}|$. On pre-legal threads $s\in \prethreads{A}$ we define:
\[
Q(s) = \sum_{i=1}^{|s|}{ar(s_i)}
\]
If $s\in \prethreads{A}$, $Q(s)$ is also the number of ways $s$ can be extended to some $sa$ (let us recall here that as a member of $\prethreads{A}$, $s$ need not be alternating): the choice of a justifier $s_i$
plus a move enabled by $s_i$. These definitions allow to express the following observation. If $\sigma: A\tto B$ is an isomorphism (thus a pre-zig-zag strategy) and $s\in \sigma$, then
$Q(s\restrict A) = Q(s\restrict B)$, because $\sigma$ being an isomorphism, it must associate each possible extension of $s\restrict A$ to an unique extension of $s\restrict B$.
But this also means that if $sab\in \sigma$ we have $Q(s\restrict A) + ar(a) = Q((s\restrict A) a) = Q((s\restrict B) b) = Q(s\restrict B) + ar(b)$, hence $ar(a) = ar(b)$. Thus to
each move $a$, $\sigma$ must associate a move with the same arity. This is a step in the right direction, but we would like a deeper connection between $a$ and $b$.

If $a\in M_A$, we will use the notation $J_a = \{m\in M_A \mid a\vdash_A m\}$.
Let us define by induction on $k$ the notion of a $k$-isomorphism between $a\in M_A$ and $b\in M_B$. 
For any $a\in M_A$ and $b\in M_B$ there is automatically a $0$-isomorphism $i_{a, b}$. A $(k+1)$-isomorphism from $a$ to $b$
is the data of an isomorphism $f : J_a \to J_b$ along with, for all $m\in J_a$, a $k$-isomorphism $f_m: m \to f(m)$.
We use the notation $m \iso_k n$ to denote the fact that there is a $k$-isomorphism from $m$ to $n$. In other words, 
we have $m\iso_k n$ if the tree of paths of length at most $k$ starting form $m$ is tree-isomorphic to the tree of paths of length at most $k$ starting from $n$. 
If $k_1 \leq k_2$, $f_1$ is a $k_1$-isomorphism and $f_2$ is a $k_2$-isomorphism, we say that $f_1$ is a prefix of $f_2$ if they agree up to depth $k_1$.
Note that in particular we have $m \iso_1 n$ if and only if $ar(m) = ar(n)$, so $m \iso_k n$ is indeed a generalization of $ar(m) = ar(n)$. By induction on $k$, one can then prove that $\sigma$ must
always associate to each move $m$ a move $n$ such that $m \iso_k n$ : to prove it for $k+1$, just apply the counting argument above on $\iso_k$-equivalence classes. From all these $k$-isomorphisms,
one can then deduce the existence of a path isomorphism between $A$ and $B$. 

This counting argument has several unsatisfying aspects, which are caused by the implicit use of the following lemma.

\begin{lem}[Slicing of bijections]
Suppose $E = E_1 + E_2$ and $F = F_1 + F_2$ are finite sets, and that $f:E\to F$ and $g: E_1\to F_1$ are bijections. Then there
is a bijection $f\setminus g: E_2 \to F_2$.
\label{slicing_isos}
\end{lem}

This lemma is obviously true by cardinality reasons.
However this proof is, computationally speaking, ``almost non-effective", in the sense that the isomorphism it produces implicitly depends on the choice of a total
ordering for $E$ and $F$. A consequence of that is that from any isomorphism in $\Gam$ we will extract an isomorphism in $\Path$, but we cannot hope
its choice to be canonical, for any reasonable meaning of ``canonical". Even worse, the witness isomorphisms given by this proof for $\iso_k$ and $\iso_{k+1}$ need not agree together.
This implies that for infinitely deep arenas, one requires König's lemma to actually build a path isomorphism from a game isomorphism. This means that we cannot deduce from the proof above
an algorithm to extract path isomorphisms.

\subsection{Extraction of a path isomorphism}

To obtain a more computationally meaningful extraction of a path iso from a game iso, we must replace the proof of Lemma \ref{slicing_isos} by
something else than counting. As formalized in the following proof, the idea is to remark that given the data of Lemma \ref{slicing_isos}, starting
from $x\in E_2$, the sequence
\begin{eqnarray*}
x_0 &=& f(x)\\
x_{n+1} &=& f\circ g^{-1}(x_n)
\end{eqnarray*}
must eventually reach $F_2$, as illustrated in Figure \ref{fig_slicing}, yielding a bijection between $E_2$ and $F_2$
(this corresponds to the construction of a \emph{trace} \cite{joyal-street-verity} on the category of finite sets and permutations).

\begin{figure}
\begin{center}
\includegraphics[scale=0.5]{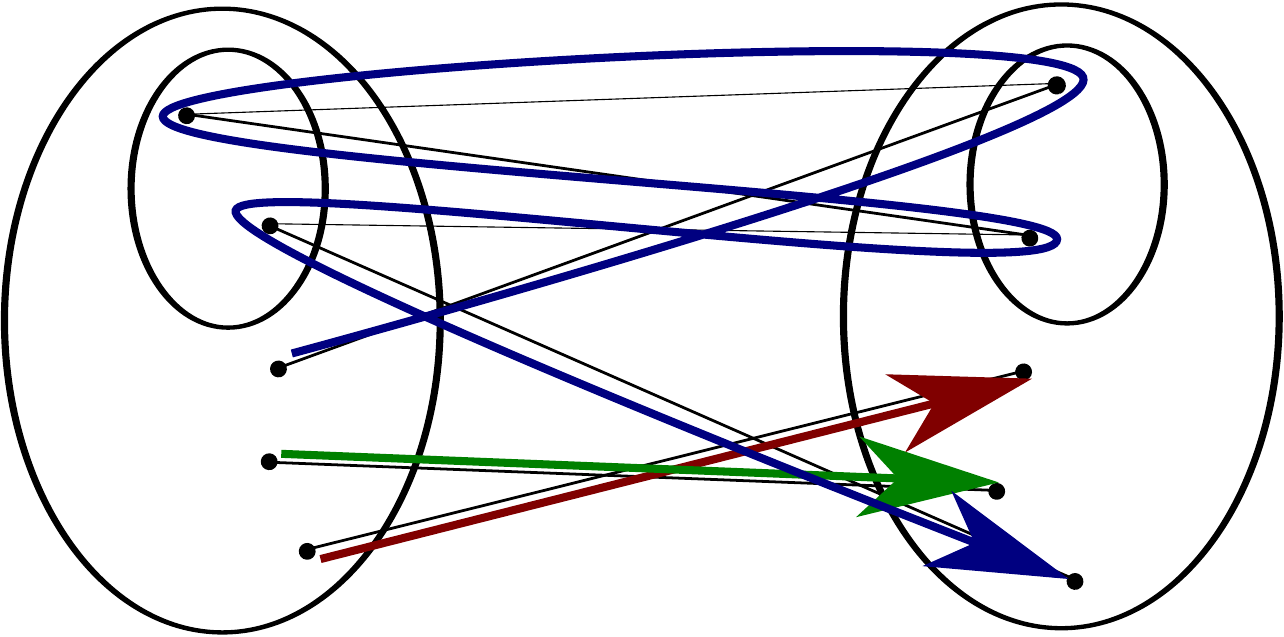}
\end{center}
\caption{Slicing of isomorphisms.}
\label{fig_slicing}
\end{figure}

\begin{prop}
If $\phi : A\to B$ is a sequential play isomorphism, then for all $sa \in \prethreads{A}$ with $\phi(sa) = \phi(s) b$, there is a family
$(h_{s, sa}^k)_{k\in \mathbb{N}}$ such that for all $k$, $h_{s, sa}^k$ is a $k$-isomorphism from $a$ to $b$. 
This family is \emph{coherent}, in the following sense: if $k_1 \leq k_2$, $h_{s,sa}^{k_1}$ is a prefix of $h_{s, sa}^{k_2}$.
\end{prop}
\begin{proof}
We will use the following notations. If $s\in \prethreads{A}$, $E_s$ will be the set of atomic extensions of $s$, that is of plays $sa\in \prethreads{A}$, and $F_s$ will be the set of
atomic extensions of $\phi(s)$. For all plays $sa\in \prethreads{A}$, although strictly speaking $E_s$ is \emph{not} a subset of $E_{sa}$, 
we have the following decomposition:
\[
E_{sa} = E_s + J_a
\]
Indeed, a move extending $sa$ can either point to some $s_i$ or to $a$.
Note also that for any $s$, $\phi: sa \mapsto \phi(s) b$ induces an isomorphism $f_s: a \mapsto b$ from $E_s$ to $F_s$.

For all $s\in \prethreads{A}$ and $sa\in E_s$, we follow the reasoning illustrated in Figure \ref{fig_slicing} and consider a bipartite directed graph $G_{s, sa}$ defined as follows: its set of vertices is $V = E_{sa} + F_{sa}$ and its set of edges is
$E = \{(x, f_{sa}(x)) \mid x\in E_{sa}\} + \{(y, f_s^{-1}(y))\mid y \in F_s\}$. This graph is ``deterministic", in the sense that the outwards degree of each vertex is at most one, moreover
the only vertices whose outwards degree is $0$ are those of $J_b$ (where $b = f_s(a)$, so $F_{sa} = F_s + J_b$). Moreover $G_{s, sa}$ must be acyclic, since $f_s$ and $f_{sa}$ are isomorphisms.
Thus from any vertex in $J_a$, there is an unique path in $G$ leading to a vertex in $J_b$; this induces an isomorphism $g_{s, sa} : J_a \to J_b$. For each pair $(m, g_{s, sa}(m))$ we
also keep track of the corresponding path $p_{s, sa}^m = (m, f_{sa}(m), f_s^{-1}(f_{sa}(m)), \dots, g_{s, sa}(m))$.

It is now time to build the $k$-isomorphisms, by induction on $k$. For $k=0$ this is obvious. For fixed $k+1\geq 1$, by induction hypothesis there is for each $sa\in \prethreads{A}$ with
$\phi(sa) = \phi(s) b$ a $k$-isomorphism $h_{s, sa}^k$ from $a$ to $b$. In particular, for fixed $sa\in \prethreads{A}$, consider the graph $G_{s, sa}$. 
Each of its edges of the form $(x, f_{sa}(x))$ are now labeled by the $k$-isomorphism $h_{sa, x}^k$ and all its edges of the form $(y, f_s^{-1}(y))$ are labeled by $(h_{s, f_s^{-1}(y)}^k)^{-1}$.
For each pair $(m, g_{s, sa}(m))$ we can now compose the labels along the path $p_{s, sa}^m$ and get a $k$-isomorphism $i_m : m \to g_{s, sa}(m)$. We then define 
$h_{s, sa}^{k+1} = (g_{s, sa}, (i_m)_{m\in J_a})$ which is as needed a $(k+1)$-isomorphism from $a$ to $b$.

Note finally that if $k_1 \leq k_2$, $h_{s, sa}^{k_1}$ is a prefix of $h_{s, sa}^{k_2}$. This is proved by simultaneous induction on $k_1$ and $k_2$. If $k_1 = 0$ this is obvious. Otherwise,
it relies on the fact that the graph $G_{s, sa}$ does not depend on $k$. Hence $h_{s, sa}^{k_1 + 1} = ( g_{s, sa}, (i_m)_{m\in J_a})$ and $h_{s, sa}^{k_2 + 1} = ( g_{s, sa}, (j_m)_{m\in J_a})$,
and each $i_m$ has be obtained from $k_1$-isomorphisms in the same way as $j_m$ has been obtained from $k_2$-isomorphisms, so it immediately boils down to the induction hypothesis.
\end{proof}

\begin{thm}
Two finitely branching arenas $A$ and $B$ are $\Gam$-isomorphic if and only if they are $\Path$-isomorphic.
\label{main}
\end{thm}
\begin{proof}
Consider an isomorphism $\sigma : A\tto B$ in $\Gam$. Restricted on plays with only two moves, it gives an isomorphism $f : I_A \to I_B$. By the previous proposition, there is for each $i\in I_A$
and for each $k\in \mathbb{N}$ a $k$-isomorphism $h^k_{\epsilon, i}: i \to f(i)$. Additionally, all these $k$-isomorphisms are compatible with each other, so they converge to an $\omega$-isomorphism
$h_{\epsilon, i}: i \to f(i)$. The iso $f$ together with $h_{\epsilon, i}$ for all $i$ define a path isomorphism from $A$ to $B$.
\end{proof}

For each pair of arenas $A, B$, we have a function $K_{A, B} : \Gam_i(A, B) \to \Path_i(A, B)$. Unfortunately,
this function fails to be a functor. Indeed, the construction is based on the more explicit proof of Lemma \ref{slicing_isos} illustrated in
Figure \ref{fig_slicing}, which is not functorial; one can easily find sets $E = E_1 + E_2$, $F = F_1 + F_2$, $G = G_1 + G_2$ along with bijections
$f_1: E \to F$, $f_2: E_1 \to F_1$, $g_1 : F \to G$ and $g_2 : F_1 \to G_1$ such that $(f\setminus f'); (g\setminus g') \neq 
(f; g)\setminus (f'; g')$, and extract from this a counter-example for the functoriality of $K_{A, B}$. However,
$K$ is a natural transformation:

\begin{prop}
The family $K_{A, B}: \Gam_i(A, B) \to \Path_i(A, B)$ is natural in $A$ and $B$, where both $\Gam_i(-, -)$ and $\Path_i(-,-)$ are seen as bifunctors
from $\Path_i^{op}\times \Path_i$ to $\Set$ (using implicitly the faithful functor from $\Path_i$ to $\Gam_i$ of Figure \ref{groupoids}).
\end{prop}
\begin{proof}
The naturality conditions expresses invariance of $K_{A, B}$ under renaming of moves in $A$ and $B$, as composing with $\Path$-isomorphisms
or $\Gam$-isomorphisms generated from $\Path$-isomorphisms only rename moves. The proof proceeds by showing that all $k$-isomorphisms
$h^k_{s, sa}$ on which the definition of $K$ relies are invariant under renaming of moves, by induction on $k$, then on $s$.
\end{proof}

\section{Syntactic isomorphisms}

\label{section_syntactic}

\subsection{Application to $\Lsums$}

Our isomorphism theorem most naturally applies to $\Gam$ (so to call-by-name languages), but $\Lsums$ is modeled in $\Fam(\Gam)_T$, 
so we have to check how our result extends to this. Let us first relate isomorphisms in $\Fam(\Gam)_T$ and isomorphisms in $\Gam$.
We start by recalling some terminology: an arena $A$ is \textbf{pointed} if it has only one initial move. A strategy
$\sigma: A \to B$ where $A$ and $B$ are pointed is \textbf{strict} if it responds to the initial move in $B$ with the initial
move in $A$, which it never plays again. Pointed arenas and strict maps form a subcategory $\Gam_\bot$ of $\Gam$. As such,
our characterisation of the isomorphisms in $\Gam$ will apply just as well on $\Gam_\bot$.

\begin{lem}
If $A$ and $B$ are isomorphic in $\Fam(\Gam)_T/\obseq$, then $TA$ and $TB$ are isomorphic in $\Gam/\obseq$.
\label{analysis_isos}
\end{lem}
\begin{proof}
It is well-known that there is a full and faithful functor from $\Fam(\Gam)_T$ to $\Gam_\bot$,
mapping $A$ to $TA$ and $f: A \to TB$ to $f^*: TA \to TB$ (assimilating the singleton family $TA$ with the arena it contains).
This functor preserves and reflects $\obseq$, so isomorphisms in $\Fam(\Gam)_T/\obseq$ correspond to isomorphisms in $\Gam_\bot/\obseq$.
They are then transfered to $\Gam/\obseq$ since it contains $\Gam_\bot/\obseq$ as a subcategory.
\end{proof}

Because of the presence of the empty type, isomorphisms in $\Gam$ do not exactly correspond to isomorphisms in $\Gam_\bot$:
unanswerable moves (as in $\intr{\unit \to 0}$) do not appear in complete plays, so $\sigma \obseq \id_A$ can do anything as soon
as one of those has been played.  If $A$ is an arena such that all questions in
$A$ are answerable (\emph{i.e.} for all $q\in M_A$ such that $\lambda^{QA}(q) = Q$, there is $a\in M_A$ such that $q \enb{A} a$ and
$\lambda^{QA}(a) = A$), we say that $A$ is \textbf{complete}. If $A$ is any arena, $\trim{A}$ is the \textbf{trimmed}
version of $A$, where we have removed the unanswerable moves along with all the moves hereditarily justified by them. Note that
for all arena $A$, $\trim{A}$ is always complete. This operation can also be applied to strategies by setting $\trim{\sigma}$
as the set of plays in $\sigma$ which do not contain unanswerable moves.

We handle the mismatch between isomorphisms in $\Gam$ and $\Gam/\obseq$ as follows:

\begin{lem}
For any arenas $A$ and $B$, $\sigma : A \to B$ and $\tau: B \to A$ form a $\Gam/\obseq$-isomorphism iff
$\trim{\sigma} : \trim{A} \to \trim{B}$ and $\trim{\tau} : \trim{B} \to \trim{A}$ form a $\Gam$-isomorphism.
\label{removeobs}
\end{lem}
\begin{proof}
Let us first note that if $\sigma: A \to B$ and $\tau: B \to C$ and $s\in \sigma; \tau$ is complete, then the witness
$u\in \sigma\parallel \tau$ must be complete as well, otherwise that would break well-bracketing. As a consequence,
$u$ contains no unanswerable move.
Hence if $\sigma$ and $\tau$ form a $\Gam/\obseq$-isomorphism, we still have $\trim{\sigma}; \trim{\tau} \obseq \id_{\trim{A}}$ and
$\trim{\tau}; \trim{\sigma} \obseq \id_{\trim{B}}$, since no unanswerable moves can arise in an interaction between $\sigma$ and $\tau$ giving rise to
a complete play. We turn now to the proof of the equivalence.\\
$\Rightarrow$.
Take $s\in \id_{\trim{A}}$. It is straightforward to see that $s$ can be completed, \emph{i.e.} there is $s'\in \id_{\trim{A}}$ such that
$s \sqsubseteq s'$ and $s'$ is complete (Opponent only plays answers, he always can because $\trim{A}$ is complete, the number of unanswered
questions decreases strictly). Therefore, $s'\in \trim{\sigma}; \trim{\tau}$, hence $s\in \trim{\sigma}; \trim{\tau}$ as well, so
$\id_{\trim{A}} \subseteq \trim{\sigma}; \trim{\tau}$. But $\id_{\trim{A}}$ is total and both strategies are deterministic, therefore
this inclusion must be an equality. The same reasoning show that $\trim{\tau}; \trim{\sigma} = \id_{\trim{B}}$ as well, so $\trim{\sigma}$ and
$\trim{\tau}$ form a $\Gam$-isomorphism.\\
$\Leftarrow$. If $\trim{\sigma}$ and $\trim{\tau}$ form a $\Gam$-isomorphism, take a complete $s\in \sigma; \tau$. As we have
proved above, the witness $u$ for $s$ does not contain any unanswerable move, hence $s\in \trim{\sigma}; \trim{\tau} = 
\id_{\trim{A}}\subseteq \id_A$. Conversely if $s\in \id_A$ is complete, then necessarily $s\in \id_{\trim{A}}$ as well. Thus,
$s\in \trim{\sigma}; \trim{\tau}$. But we have seen above that by necessity the witness $u\in \trim{\sigma}\parallel \trim{\tau}$
is complete as well and as such cannot contain any unanswerable move, so $s\in \sigma; \tau$ and $\sigma; \tau \obseq \id_A$.
\end{proof}

The results above allow to prove that isomorphisms in $\Lsums$ yield $\Gam$-isomorphisms, hence $\Path$-isomorphism by an application
of Theorem \ref{main}. It remains to show that types that give rise to $\Path$-isomorphic arenas are characterized by the equational theory
$\mathcal{E}$. For this purpose, it will be convenient to start by putting types in \emph{canonical form}, as described below. 

\begin{lem}[Canonical form]
Any type of $\Lsums$ has a representative (up to $\mathcal{E}$) generated by $T$ in, with $|I|\geq 2$.
\begin{eqnarray*}
T &::=& 0 \mid 1 \mid S \mid P \mid A\\
S &::=& \Sigma_{i\in I} L\\
P &::=& \Pi_{i\in I} A\\
A &::=& L \tto R\\
L &::=& A \mid P \mid 1\\
R &::=& A \mid P \mid S \mid 1 
\end{eqnarray*}
\end{lem}
\begin{proof}
First eliminate all occurrences of $\var$ using the last equation of $\mathcal{E}$.
We make the rest of $\mathcal{E}$ into a rewriting system by directing the equations from left to right, removing those for commutativity,
adding an expansion $(A + B)\times C \leadsto A\times C + B \times C$, and right cancellation of units. It is then straightforward to prove
that the following measure strictly decreases with each reduction: $|0| = |\unit| = 1$, $|A + B| = |A| + 2|B|$, $|A\times B| = (|A|+1)|B|$
and $|A\to B| = (|B|+1)^{|A|}$. It is then a simple induction to find a derivation tree from $T$ for types that are normal forms for this reduction.
\end{proof}

\begin{lem}
Let us extend $\mathrm{trim}$ to families by setting $\trim{(A_i)_{i\in I}} = (\trim{A_i})_{i\in I}$.
For any type $B$ in canonical form, we have $\trim{\intr{B}} = \intr{B}$. Moreover, we have the following equivalences:
\begin{enumerate}[\em(1)]
\item $B = 0$ iff $\intr{B} = 0$,
\item $B = 1$ iff $\intr{B} = 1$,
\item $B$ is generated by $S$ iff $\intr{B}$ has at least two members,
\item $B$ is generated by $P$ iff $\intr{B} = \{B'\}$ where $B'$ has at least two initial moves,
\item $B$ is generated by $A$ iff $\intr{B} = \{B'\}$ where $B'$ has exactly one initial move.
\end{enumerate}
\label{dichotomy}
\end{lem}
\begin{proof}
Straightforward.
\end{proof}

\begin{prop}
If $\trim{T\intr{A}}$ and $\trim{T\intr{B}}$ are $\Path$-isomorphic, then $A \iso_{\mathcal{E}} B$.
\label{equational_isos}
\end{prop}
\begin{proof}
First, note that $\trim{T\intr{A}} = 1$ if $\intr{A}$ is the empty family and $T(\trim{A_i})$ otherwise.
We reason by simultaneous induction on $A$ and $B$, that we both suppose in canonical form. By Lemma \ref{dichotomy} and the remark above,
this means that we get rid of $\mathrm{trim}$ and suppose $T\intr{A}$ and $T\intr{B}$ to be $\Path$-isomorphic.
Clearly, $\intr{A}$ and $\intr{B}$ must be in the same
case of Lemma \ref{dichotomy} (otherwise it is easily checked that they cannot be isomorphic). If it is case (1) (resp. (2)), then both
$A$ and $B$ have $0$ (resp. $1$) as canonical form and $A \iso_{\mathcal{E}} B$.

If it is case (3), then $A = \Sigma_{i\in I} A_i$ and $B = \Sigma_{j\in J} B_j$, with $\intr{A} = (\intr{A_i})_{i\in I}$ and
$\intr{B} = (\intr{B_j})_{j\in J}$. The $\Path$-isomorphism between $T\intr{A}$ and $T\intr{B}$ yields a bijection
$f: I \to J$ and for all $i\in I$ a $\Path$-isomorphism $\phi_i : \intr{A_i} \to \intr{B_{f(i)}}$.
By induction hypothesis, this means that for all $i\in I$ we have $A_i \iso_{\mathcal{E}} B_{f(i)}$. By repeated
uses of commutativity and associativity of $+$, we conclude that $A \iso_{\mathcal{E}} B$. 

If it is case (4), then 
$A = \Pi_{i\in I} A_i$ and $B = \Pi_{j\in J} B_j$. Then $\intr{A} = \{\Pi_{i\in I} A_i\}$ and 
$\intr{B} = \{\Pi_{j\in J} B_j\}$. 
Then, the $\Path$-isomorphism between $T\intr{A}$ and $T\intr{B}$ yields a $\Path$-isomorphism between $\Pi_{i\in I} A_i$ and $\Pi_{j\in J} B_j$.
In turn, this yields a bijection $f: I \to J$, and
for each $i\in I$ a $\Path$-isomorphism between $A_i$ and $B_{f(i)}$. By induction hypothesis, this means that for all $i\in I$
we have $A_i \iso_{\mathcal{E}} B_{f(i)}$, therefore $A \iso_{\mathcal{E}} B$ by repeated uses of associativity and commutativity for $\times$.

If it is case (5), then $A = A_1 \to A_2$ and $B = B_1 \to B_2$. Both $A_1$ and $B_1$ are generated by $L$, so they must consist of singleton
families $\{A'_1\}$ and $\{B'_1\}$. Then, $\intr{A} = \{A'_1 \tto T\intr{A_2}\}$ and $\intr{B} = \{B'_1 \tto T\intr{B_2}\}$, and the
$\Path$-isomorphism between $T\intr{A}$ and $T\intr{B}$ yields a $\Path$-isomorphism $\phi$ between $A'_1 \tto T\intr{A_2}$ and
$B'_1 \tto T\intr{B_2}$. Since $\phi$ preserves Q/A labelling, it decomposes into $\Path$-isomorphisms 
$\phi_1 : A'_1 \to B'_1$ and $\phi_2 : T\intr{A_2} \to T\intr{B_2}$. By induction hypothesis this implies that
$A_1 \iso_{\mathcal{E}} B_1$ and $A_2 \iso_{\mathcal{E}} B_2$, thus $A \iso_{\mathcal{E}} B$.
\end{proof}

Putting all of these together:

\begin{thm}
For any types $A, B$ of $\Lsums$, we have the following equivalence:
\[
A \iso_{\Lsums} B \Leftrightarrow A \iso_{\mathcal{E}} B
\]
\end{thm}
\begin{proof}
Suppose we have a (syntactic) isomorphism $x : A \vdash M: B$ and $y: B \vdash N : A$. It then easy to check that
$\intr{N \circ M} = \intr{M}; \intr{N}$, when the former composition is syntactic composition and the latter composition in $\Fam(\Gam)_T$.
Likewise, we have $\intr{x : A \vdash x:A} = \id_{\intr{A}}$ (identity in $\Fam(\Gam)_T$). By full abstraction, we have
$\intr{M}; \intr{N} \obseq \id_{\intr{A}}$ and $\intr{N}; \intr{M} \obseq \id_{\intr{B}}$, so we have a $\Fam(\Gam)_T/\obseq$-isomorphism
between $\intr{A}$ and $\intr{B}$. By Lemma \ref{analysis_isos}, this means that $T\intr{A}$ and $T\intr{B}$ are $\Gam/\obseq$-isomorphic.
By Lemma \ref{removeobs}, $\trim{T\intr{A}}$ and $\trim{T\intr{B}}$ are $\Gam$-isomorphic. By Theorem \ref{main}, they are
$\Path$-isomorphic. By Proposition \ref{equational_isos}, this implies that $A \iso_{\mathcal{E}} B$. 

Conversely, it is straightforward to check that all equations in $\mathcal{E}$ between $A$ and $B$ give rise to $\Path$-isomorphisms between
$\trim{T\intr{A}}$ and $\trim{T\intr{B}}$. By Laurent's theorem (the isomorphism of groupoids $H' : \Inn_i \to \Path_i$, see Section 
\ref{laurentstheorem}), there is an innocent isomorphism $\sigma: \trim{T\intr{A}} \to \trim{T\intr{B}}$,
$\tau: \trim{T\intr{B}} \to \trim{T\intr{A}}$, note that $\sigma$ and $\tau$ have finite view functions. We also have 
$\sigma : T\intr{A} \to T\intr{B}$ and $\tau: T\intr{B} \to T\intr{A}$, although they might not form an isomorphism anymore. However, they
do form a $\Gam/\obseq$-isomorphism by Lemma \ref{removeobs}. By construction they are strict, so they come from morphisms
$\sigma' : \intr{A} \to T\intr{B}$ and $\tau': \intr{B} \to T\intr{A}$ forming an isomorphism in $\Fam(\Gam)_T/\obseq$. By innocent
definability, there are $x: A \vdash M: B$ and $y: B \vdash N: A$ such that $\intr{M} = \sigma'$ and $\intr{N} = \tau'$. By full abstraction,
$M$ and $N$ must form a syntactic isomorphism of types.
\end{proof}

\subsection{Isomorphisms in the presence of $\mathtt{nat}$}

Consider the programming language $\Lang$ from \cite{ahm}, obtained from $\Lsums$ by replacing sums by $\bool$ and $\nat$, along
with the associated combinators. As proved in \cite{ahm}, this language has a fully abstract interpretation in $\BFam(\biggam)_T$,
where $\biggam$ is the category of not necessarily finitely branching arenas, and single-threaded strategies.

As suggested by the importance of counting in the proof, the presence of $\mathtt{nat}$ makes it possible to build new isomorphisms
by playing Hilbert's hotel. Of course there are obvious new isomorphisms, such as $\nat \iso \nat + \nat$ or $\nat \iso \nat \times \nat$,
which are realizable by purely functional terms. What is less obvious is that in the presence of higher-order state, one can define new
isomorphisms which did not exist in the purely functional fragment of $\Lang$. In this section, we will detail as much as possible
one of those new isomorphisms, then mention a few others.

Our main example will be an isomorphism between $\nat \to \nat \to \unit$ and $\nat \to \unit$. Although this seems to follow from 
$\nat \times \nat \iso \nat$, this is not the case since curryfication is in general not a valid isomorphism in a call-by-value
language. As a consequence of Laurent's theorem, no purely functional isomorphism can exist between these two types because their
corresponding arenas are not tree-isomorphic.

\begin{prop}
There is an isomorphism in $\BFam(\biggam)_T$ between $\intr{\nat \to \nat \to \unit}$ and $\intr{\nat \to \unit}$.
\end{prop}
\begin{proof}
By definition of the interpretation of types, this boils down to an isomorphism in $\biggam$ between the two arenas
$T\intr{\nat \to \nat \to \unit}$ and $T\intr{\nat \to \unit}$ represented in Figure \ref{isoar}. Informally, the
left-to-right isomorphism can be described as follows:

As long as no $b$ has been played, it behaves as the identity. The first time a $b$ is played, Player copies it on
the right side. One can then check that the play has $\mathbb{N} + \mathbb{N}$ possible extensions on the left hand side, whereas
it only has $\mathbb{N}$ extensions on the right hand side. Therefore, Player has to fix a bijection $\phi: \mathbb{N} + \mathbb{N} \to \mathbb{N}$
and play accordingly. In general if $n$ occurrences of $b$ have been played, there are $n\mathbb{N}$ $q_i$s available on the
left hand side and still $\mathbb{N}$ on the right hand side, therefore Player has to follow a bijection $\phi_n : n\mathbb{N} \to \mathbb{N}$.

Thus there is in fact an infinity of different isomorphisms between $T\intr{\nat \to \nat \to \unit}$ and $T\intr{\nat \to \unit}$, one
for each family $(\phi_n)_{n\in \mathbb{N}}$ of bijections between $n \mathbb{N}$ and $\mathbb{N}$.
\end{proof}

\begin{figure}
\[
\xymatrix@C=5pt@R=10pt{
&&&&q	\ar@{-}[d]\\
&&&&a	\ar@{-}@/_/[dlll]
	\ar@{-}[d]
	\ar@{-}@/^/[drr]\\
&q_1	\ar@{-}[d]&&&q_2\ar@{-}[d]&&\dots\\
&b	\ar@{-}@/_/[dl]
	\ar@{-}[d]
	\ar@{-}@/^/[dr]&&&
b      \ar@{-}@/_/[dl]
        \ar@{-}[d]
        \ar@{-}@/^/[dr]\\
q_1	\ar@{-}[d]&
q_2	\ar@{-}[d]&
\dots&
q_1	\ar@{-}[d]&
q_2	\ar@{-}[d]&
\dots\\
c&c&&c&c
}
~~~~
\xymatrix@C=5pt@R=10pt{
&q	\ar@{-}[d]\\
&a	\ar@{-}@/_/[dl]
	\ar@{-}[d]
	\ar@{-}@/^/[dr]\\
q_1	\ar@{-}[d]&
q_2	\ar@{-}[d]&
\dots\\
b&b
}
\]
\caption{Non-trivially isomorphic arenas in $\biggam$}
\label{isoar}
\end{figure}

We note that the strategy from $\nat \to \unit$ to $\nat \to \nat \to \unit$ is visible, so this also gives an example of a morphism in $\Vis$
which is not invertible in $\Vis$ but becomes invertible in $\Gam$.
These strategies are not compact so the definability theorem does not apply, however we can nonetheless manually extract corresponding programs from
them. We display them in Figure \ref{coolisos}, where we suppose that a family of bijections $\phi_n : n\mathbb{N} \to \mathbb{N}$ has already
been defined. Unfortunately, these terms are too complex to hope for a reasonably-sized direct
proof that their interpretations give the strategies described above or even that they form an isomorphism. This kind of difficulty emphasizes the
need for new algebraic methods to manipulate and prove properties of imperative higher-order programs.

\begin{figure*}
{\small
\begin{minipage}{0.49\linewidth}
\[
\begin{array}{l}
f: \nat \to \nat \to \unit \vdash\\ 
~~\mathtt{new}~\mathtt{count} := 0,~\mathtt{func} := \bot~\mathtt{in}\\
~~\lambda n.~\mathtt{let}~(p, q)~=~\phi_{!count + 1}^{-1}(n)~\mathtt{in}\\
~~~~~~~\mathtt{if}~p = 0~\mathtt{then}\\
~~~~~~~~~\mathtt{let}~x = f~q~\mathtt{in}\\
~~~~~~~~~\mathtt{count}:=~!\mathtt{count} + 1;\\
~~~~~~~~~\mathtt{let}~c =~!\mathtt{count}~\mathtt{in}\\
~~~~~~~~~\mathtt{func} :=~(\mathtt{let}~g =~!\mathtt{func}~\mathtt{in}\\ 
~~~~~~~~~~~(\lambda n.~\mathtt{if}~n =~c~\mathtt{then}~x\\
~~~~~~~~~~~~~~~~~~\mathtt{else}~g~n))\\
~~~~~~~\mathtt{else}~!\mathtt{func}~p~q
\end{array}
\]
\end{minipage}
\begin{minipage}{0.49\linewidth}
\[
\begin{array}{l}
f: \nat \to \unit \vdash\\
~~\mathtt{new}~\mathtt{count} :=~0~\mathtt{in}\\
~~\lambda n.~f~(\phi_{!\mathtt{count} + 1}(0,n));\\
~~~~~~~\mathtt{count}:=~!\mathtt{count} +1;\\
~~~~~~~\mathtt{let}~c =~!\mathtt{count}~\mathtt{in}\\
~~~~~~~\lambda p.~f~(\phi_{!\mathtt{count} + 1}(c,p))
\end{array}
\]
\end{minipage}
}
\caption{Type isomorphism in $\Lang$ between $\mathtt{nat} \to \mathtt{nat} \to \unit$ and $\mathtt{nat} \to \unit$.}
\label{coolisos}
\end{figure*}

It seems difficult to characterize exactly the new isomorphisms that natural numbers allow to define. One can prove that
the types $(\nat \to \unit) \to (\nat \to \unit) \to \unit$ and $(\nat \to \unit) \to (\unit \to \unit) \to \unit$ are
isomorphic, showing that isomorphisms are non-local. Even worse, replacing any occurrence of $\unit$ by $\bool$ in the
types above yields non-isomorphic types. Likewise, $\nat \to \nat \to \bool$ and $\nat \to \bool$ are not isomorphic.

It is also interesting to note that composing $\nat \to \nat \to \unit \iso \nat \to \unit$ with $\nat \times \nat \iso \nat$ provides an isomorphism
$\nat \to \nat \to \unit \iso \nat \times \nat \to \unit$, even though curryfication is not a valid isomorphism in general. However one
should keep in mind that the terms realizing this isomorphism have nothing in common with curryfication, as they have
to use higher-order references in a non-trivial way. In particular, it seems unlikely that they can be used for modularity
purposes, putting some limits to the idea that isomorphisms of types always provide the good notion of equivalence on which
programmers should rely.

\section{Conclusion}

We solved Laurent's conjecture and characterized the isomorphisms of types in $\Lsums$. Surprisingly, we realized that the combination of higher-order
references, natural numbers and call-by-value allowed to define new non-trivial type isomorphisms. Note however that if well-bracketing is satisfied, the proof
of our core game-theoretic theorem adapts directly to arenas where all moves only enable a finite number of questions, but an
arbitrary numbers of answers. As a consequence, there are no non-trivial isomorphisms (\emph{i.e.} not already present in the $\lambda$-calculus) in the call-by-name variant of $\Lang$, although
we can define one using \texttt{call}/\texttt{cc}.

Note that despite the seemingly restricted power of $\Lsums$, our theorem does apply to all real-life programming languages that have a bounded type of integer, 
such as $\mathtt{bool}^{32}$ or $\mathtt{bool}^{64}$: in this setting, no non-trivial isomorphism can exist. However unbounded natural numbers can be defined using recursive types, so
the isomorphism above can be implemented in a call-by-value programming language with recursive types and general references, such as \textsc{Ocaml}.

\textit{Acknowledgments.} We would like to thank Guy McCusker and Nikos Tzevelekos for stimulating discussions
about the new non-trivial isomorphisms. 

\bibliographystyle{plain}
\bibliography{lics2011}

\end{document}